\documentclass[a4paper,11pt,aps,pra,notitlepage,superscriptaddress,nofootinbib,tightenlines]{revtex4-1}

\usepackage{amsmath}
\usepackage{amsthm}
\usepackage{amsfonts}
\usepackage{mathrsfs}
\usepackage{microtype}
\usepackage{enumitem}

\usepackage[colorlinks=true,linkcolor=blue,citecolor=blue,plainpages=false,pdfpagelabels]{hyperref}
\usepackage{color}

\DeclareMathOperator{\spec}{spec}
\DeclareMathOperator{\tr}{tr}
\DeclareMathOperator{\Exp}{\mathbb{E}}
\DeclareMathOperator*{\argmin}{\arg\min}
\DeclareMathOperator*{\argmax}{\arg\max}

\newcommand{\eps}{\varepsilon}

\newcommand{\cS}{\mathcal{S}}
\newcommand{\cH}{\mathscr{H}}
\newcommand{\cF}{\mathcal{F}}
\newcommand{\Herm}{\mathcal{H}}
\newcommand{\Pos}{\mathcal{P}}
\newcommand{\Unit}{\mathcal{U}}
\newcommand{\cP}{\mathcal{P}}
\newcommand{\id}{{1}}

\newtheorem{lemma}{Lemma}
\newtheorem{proposition}[lemma]{Proposition}
\newtheorem{theorem}[lemma]{Theorem}
\newtheorem{corollary}[lemma]{Corollary}

\numberwithin{equation}{section}

\begin{document}

\title{\Large Correlation Detection and an Operational Interpretation of the R\'enyi Mutual Information}

\author{Masahito Hayashi}
\affiliation{Graduate School of Mathematics, Nagoya University, 
Furocho, Chikusaku, Nagoya, 464-860, Japan}
\affiliation{Centre for Quantum Technologies, National University of Singapore, Singapore 117543, Singapore}
\author{Marco Tomamichel}
\affiliation{School of Physics, The University of Sydney, Sydney 2006, Australia}
\affiliation{Centre for Quantum Technologies, National University of Singapore, Singapore 117543, Singapore}

\begin{abstract}
A variety of new measures of quantum R\'enyi mutual information and quantum R\'enyi conditional entropy have recently been proposed, and some of their mathematical properties explored. Here, we show that the R\'enyi mutual information attains operational meaning in the context of composite hypothesis testing, when the null hypothesis is a fixed bipartite state and the alternative hypothesis consists of all product states that share one marginal with the null hypothesis. This hypothesis testing problem occurs naturally in channel coding, where it corresponds to testing whether a state is the output of a given quantum channel or of a ``useless'' channel whose output is decoupled from the environment. Similarly, we establish an operational interpretation of R\'enyi conditional entropy by choosing an alternative hypothesis that consists of product states that are maximally mixed on one system. Specialized to classical probability distributions, our results also establish an operational interpretation of R\'enyi mutual information and R\'enyi conditional entropy.
\end{abstract}

\maketitle



\section{Introduction}
\label{sc:intro}

In order to distill useful measures of R\'enyi mutual information and R\'enyi conditional entropy from a plethora of possible definitions, it is important to find out which definitions correspond to relevant operational quantities. 
As such an operational task, let us consider how efficiently an arbitrary bipartite correlated quantum state $\rho_{AB}$ on systems $A$ and $B$
can be distinguished from product states
when the marginal of $\rho_{AB}$ on $A$ is known to be $\rho_A$.
This problem can be regarded as the problem of detecting correlations in the state $\rho_{AB}$.
More precisely, we want to consider the minimal probability that we
erroneously select the state $\rho_{AB}$ 
when the actual state is a product state of the form
$\rho_A \otimes \sigma_B$,
under a constraint for the opposite error. 
This problem can be studied as
the Hoeffding bound and Stein's lemma for the following sequence of binary \emph{composite hypothesis testing} problems for $n \in \mathbb{N}$:
\begin{description}[itemsep=0mm]
  \item[null hypothesis] The state is $\rho_{AB}^{\otimes n}$.
  \item[alternative hypothesis] The state is of the form $\rho_A^{\otimes n} \otimes \sigma_{B^n}$ with $\sigma_{B^n}$ any state on $B^{\otimes n}$.
\end{description}
The literature on quantum hypothesis testing is mostly focused on hypothesis testing against a simple hypothesis~\cite{ogawa04,ogawa00,nussbaum09,noetzel13,nagaoka07,nagaoka01,nagaoka06,mosonyiogawa13,mosonyi14,audenaert07,tomamichel12,li12,hiai91,hayashi07}, except for the papers \cite{kumagai13,hayashi06,brandao10}.

The above hypothesis test figures prominently when analyzing various channel coding questions in classical~\cite{polyanskiy10,hayashi09,fabregas16} as well as quantum information processing~\cite{hayashi03}. This connection is particularly important when analyzing how much information can be transmitted with a single use of a quantum channel~\cite{wang10,matthews12} or when approximating how much information can be transmitted with finitely many uses of the channel~\cite{tomamicheltan14,datta14}.
There, the problem is specified by a description of a channel $\mathcal{E}_{A' \to B}$ and a bipartite state $\rho_{AA'}$ where the system $A$ constitutes an environment of the channel, $A'$ is the channel input, and $B$ its output. We are given an unknown state on $n$ copies of $A$ and $B$ and consider the following two hypotheses.
\begin{description}[itemsep=0mm]
  \item[null hypothesis] The state is the output of $n$ uses of the channel $\mathcal{E}_{A'\to B}$, namely the state is exactly $\rho_{AB}^{\otimes n}$ where $\rho_{AB} := \mathcal{E}_{A' \to B}[\rho_{AA'}]$.
  \item[alternative hypothesis] The state is the output of a ``useless'' channel and decoupled from the environment, namely it is of the form $\rho_A^{\otimes n} \otimes \sigma_{B^n}$ with $\sigma_{B^n}$ any state on $n$ copies of $B$.
\end{description}
Polyanskiy~\cite[Sec.~II]{polyanskiy13} discusses the classical special case of this hypothesis testing problem.

A \emph{hypothesis test} for this problem is a binary positive operator-valued measure $\{ Q_{A^nB^n}, \id_{A^nB^n} - Q_{A^nB^n} \}$ on the $n$ copies of the systems $A$ and $B$, 
determined by an operator $0 \leq Q_{A^nB^n} \leq \id_{A^nB^n}$. If the first event (corresponding to $Q_{A^nB^n}$) occurs we select the null hypothesis, and in case the second event (corresponding to $\id_{A^nB^n} - Q_{A^nB^n}$) occurs we select the alternative hypothesis. 
The \emph{error of the first kind}, $\alpha_n(Q_{A^nB^n})$, is defined as the probability with which we wrongly conclude that the alternative hypothesis is correct even if the state is $\rho_{AB}^{\otimes n}$, given by
\begin{align}
  \alpha_n(Q_{A^nB^n}) = \tr[\rho_{AB}^{\otimes n} (\id_{A^nB^n} - Q_{A^nB^n})] .
\end{align}
Conversely, \emph{the error of the second kind}, $\beta_n(Q_{A^nB^n})$, is defined as the probability with which we wrongly conclude that the null hypothesis is correct even if the state is of the form $\rho_A^{\otimes n} \otimes \sigma_{B^n}$ for some $\sigma_{B^n}$, given by
\begin{align}
  \beta_n(Q_{A^nB^n}) = \max_{\sigma_{B^n}} \tr[ \rho_A^{\otimes n} \otimes \sigma_{B^n} \, Q_{A^nB^n}] ,
\end{align}
where the maximum is taken over all states $\sigma_{B^n}$ on $n$ copies of $B$.

\paragraph*{Main Results.}
The main contribution of this paper is an asymptotic analysis of the fundamental trade-off between these two errors as $n$ goes to infinity.
To investigate this trade-off, we ask the following questions: let us assume that our test is such that 
$\beta_n(Q_{A^nB^n}) \leq \exp(-nR)$, what is the minimum value of $\alpha_n(Q_{A^nB^n})$ we can achieve? The answer is different depending on whether $R$ is smaller or larger than the \emph{mutual information} between $A$ and $B$, denoted $I(A\!:\!B)_{\rho}$. If $R < I(A\!:\!B)_{\rho}$, we show that the minimal error of the first kind vanishes exponentially fast in $n$. This implies a \emph{quantum Stein's lemma}~\cite{hiai91} for the above composite hypothesis testing problem. 

More formally, we define
\begin{align}
  \hat{\alpha}_n(nR) = 
\min_{0 \leq Q_{A^nB^n} \leq \id} \Big\{ \alpha_n(Q_{A^nB^n}) \,\Big|\, \beta_n(Q_{A^nB^n}) \leq \exp(-nR) \Big\} 
\end{align}
and investigate the exact exponent with which this error vanishes as $n$ goes to infinity, yielding a \emph{quantum Hoeffding bound}~\cite{hayashi07,nagaoka06} for our composite hypothesis testing problem. We find that the exponents are determined by the \emph{R\'enyi mutual information}, defined as
\begin{align}
  I_{\alpha}(A\!:\!B)_{\rho} = \min_{\sigma_B} D_{\alpha}(\rho_{AB}\|\rho_A \otimes \sigma_B) , \qquad \textrm{for} \qquad \alpha \in (0, 1) ,
\end{align}
where $
D_{\alpha}(\rho\|\sigma) := \frac{1}{\alpha-1} \log \tr \big[ \sigma^{\frac{1-\alpha}2}\rho^{\alpha} \sigma^{\frac{1-\alpha}2} \big]$ 
is the R\'enyi relative entropy first investigated by Petz (see, e.g.~\cite{ohya93}) and the minimization is over all states $\sigma_B$ on $B$.
We obtain 
  \begin{equation}
  \label{eq:hoeffding-thm4}
    \lim_{n \to \infty} \left\{ - \frac{1}{n} \log \hat{\alpha}_n(nR) \right\} 
= \sup_{s \in (0, 1)} \left\{ \frac{1-s}{s} 
\big(   I_{s}(A\!:\!B)_{\rho} - R \big) \right\}.
  \end{equation}

On the other hand, if $R > I(A\!:\!B)_{\rho}$, we show that $\hat{\alpha}_n(n R)$ must approach one exponentially fast in $n$. This implies the strong converse for quantum Stein's lemma~\cite{ogawa00} for our problem. We then find the exact exponents (also called \emph{strong converse exponenents}, see~\cite[Ch.~3]{hayashi06} and~\cite{ogawa00,mosonyiogawa13}) with which the error of the first kind goes to one as $n$ goes to infinity and we find that in our case the exponent is determined by the
\emph{sandwiched R\'enyi mutual information}~\cite{beigi13,gupta13}, given as
\begin{align}
  \widetilde{I}_{\alpha}(A\!:\!B)_{\rho} = \min_{\sigma_B} \widetilde{D}_{\alpha}(\rho_{AB}\|\rho_A \otimes \sigma_B) , \qquad \textrm{for} \qquad \alpha > 1 ,
\end{align}
where $\widetilde{D}_{\alpha}(\rho\|\sigma) 
:= \frac{1}{\alpha-1} \log \tr \big[ \big( \sigma^{\frac{1-\alpha}{2\alpha}} 
\rho \sigma^{\frac{1-\alpha}{2\alpha}}
\big)^{\alpha}  \big]$ 
is the (sandwiched) R\'enyi divergence~\cite{lennert13,wilde13}. 
We obtain 
  \begin{equation}
    \lim_{n \to \infty} \left\{ - \frac{1}{n} \log \big( 1 - \hat{\alpha}_n(nR) \big) \right\} = \sup_{s > 1} \left\{ \frac{s-1}{s} \left( R - \widetilde{I}_{s}(A:B)_\rho \right) \right\}. 
  \end{equation}

Hence, we show that the above composite hypothesis testing problem yields an operational interpretation
for different definitions of the R\'enyi mutual information for the two ranges of $\alpha$, paralleling the observation in~\cite{mosonyiogawa13}.

Finally, we also perform a second order analysis for quantum Stein's lemma~\cite{tomamichel12,li12} and show that the minimal error of the first kind converges to a constant if $\beta_n(Q_{A^nB^n}) \leq \exp(-n I(A\!:\!B)_{\rho} - \sqrt{n} r)$ for some $r \in \mathbb{R}$.
Then, for any $r \in \mathbb{R}$, we have
  \begin{align}
\lim_{n \to \infty} \left\{ \hat{\alpha}_n\big( n I(A\!:\!B)_\rho 
+ \sqrt{n}\,r \big) \right\} = \Phi \left( \frac{r}{\sqrt{V(A\!:\!B)_\rho}} \right) ,
  \end{align}
  where $\Phi$ is the cumulative standard normal (Gaussian) distribution
and
\begin{align}
  V(A\!:\!B)_\rho := \tr \Big[ \rho_{AB} \big(\log \rho_{AB} - \log \rho_{A} \otimes \rho_{B}
- I(A\!:\!B)_\rho\big)^2 \Big].
\end{align}
is the \emph{mutual information variance}.

Analogously, an operational interpretation for \emph{conditional R\'enyi entropies} is established by considering the following binary hypotheses testing problem, which is motivated by the task of decoupling of quantum states. The problem is specified by a description of a state $\rho_{AB}$. Given an unknown state on $A$ and $B$, consider the following two hypotheses:
\begin{description}[itemsep=0mm]
  \item[null hypothesis] The state is the $n$-fold product of $\rho_{AB}$, namely $\rho_{AB}^{\otimes n}$.
  \item[alternative hypothesis] The state is uniform on $A^n$ and decoupled form $B^n$, i.e.\ it is of the form $\pi_A^{\otimes n} \otimes \sigma_{B^n}$, where $\pi_A$ is the fully mixed state on $A$.
\end{description}
The same analysis as above applied to this problem reveals that the exponents in the quantum Hoeffding bound are determined by the R\'enyi conditional entropies defined as~\cite{tomamichel13}
\begin{align}
  H_{\alpha}^{\uparrow}(A|B)_{\rho} = - \min_{\sigma_B} D_{\alpha}(\rho_{AB}\|\id_A \otimes \sigma_B), \qquad \textrm{for} \qquad \alpha \in (0, 1) \,,
\end{align}
and the strong converse exponents are determined by the \emph{sandwiched} conditional R\'enyi entropies~\cite{lennert13}
\begin{align}
    \widetilde{H}_{\alpha}^{\uparrow}(A|B)_{\rho} = - \min_{\sigma_B} \widetilde{D}_{\alpha}(\rho_{AB}\|\id_A \otimes \sigma_B), \qquad \textrm{for} \qquad \alpha > 1 \,.
\end{align}

%
\paragraph*{Related Work.}
Complementary and concurrent to this work, Cooney \emph{et al.}~\cite{cooney14} investigated the strong converse exponents for a similar hypothesis testing problem when adaptive strategies are allowed\,---\,however, they did not treat the case of a composite alternative hypothesis and they also did not analyze the error exponents in the quantum Hoeffding bound. 

Our proof of the strong converse exponents parallels the development in a very recent preprint by Mosonyi and Ogawa~\cite{mosonyi14}. There, the authors consider correlated states and use the G\"artner-Ellis theorem of classical large deviation theory in order to investigate the asymptotic error exponents in the presence of correlations. Here, we are not interested in correlated states \emph{per se}, 
but our proof technique based on pinching naturally leads us to a classical hypothesis testing problem with correlated distributions, for which the G\"artner-Ellis theorem again provides the right solution.
The main difference from their paper is the composite alternative hypothesis, which is not discussed in \cite{mosonyi14}.
To derive the strong converse exponents even for the composite alternative hypothesis, 
we combine the pinching method and the irreducible decomposition.
This idea was firstly initiated by \cite{hayashi97} and was developed in \cite{hayashi02b,tomamichel12}.

\paragraph*{Outline.}
The remainder of this paper is structured as follows. In Section~\ref{sc:pre} we introduce the necessary notation and mathematical preliminaries, and we discuss some properties of the R\'enyi divergence. We believe that Lemma~\ref{lm:pinching} and Corollary~\ref{cor:achieve} may be of independent interest. In Section~\ref{sc:mui} we define the generalized R\'enyi mutual information (which formally generalizes both R\'enyi mutual information and R\'enyi conditional entropy) and discuss various properties, including a duality relation and additivity. Most importantly, in Proposition~\ref{pr:universal}, we show that it can be represented as an asymptotic limit of classical R\'enyi divergences. 

Then, in Section~\ref{sc:prob} we formally define the composite hypothesis test we consider and the required operational quantities. In doing so, we introduce a slightly more general problem that includes the two hypothesis testing problems discussed previously as special cases.
 In Section~\ref{sc:hoeffding} we prove an analogue of the quantum Hoeffding bound that establishes the operational meaning for R\'enyi mutual information and R\'enyi conditional entropy for $\alpha < 1$. Moreover, in Section~\ref{sc:sc} we find the strong converse exponents for our problem, yielding an operational meaning for the R\'enyi mutual information and R\'enyi conditional entropy for $\alpha > 1$. As in the non-composite case, for $\alpha > 1$ the relevant R\'enyi divergence is the ``sandwiched'' R\'enyi divergence. 
 
 We conclude our treatment of this problem by considering the second order asymptotics in~Section~\ref{sc:second}. This section is interesting on its own since it provides a new and more intuitive proof of the achievability of the second order that also easily adapts to non-composite hypothesis testing. 


\section{Notation and Preliminaries}
\label{sc:pre}

We model \emph{quantum systems}, denoted by capital letters (e.g., $A$, $B$), by finite-dimensional Hilbert spaces (e.g., $\cH_A$, $\cH_B$). Moreover, $A^n$ denotes a quantum system composed of $n$ copies of the system $A$, modeled by an $n$-fold tensor product of Hilbert spaces, $\cH_{A^n} = \cH_A^{\otimes n}$. We denote by $\Unit(A)$, $\Herm(A)$ and $\Pos(A)$ the set of unitary, Hermitian, and positive semi-definite operators acting on $\cH_A$, respectively. We denote the identity operator on $\cH_A$ by $\id_A$ and the partial trace by $\tr_A$. Furthermore, we use $|A|$ to denote the dimension of the Hilbert space $\cH_A$. Finally for $L, R \in \Pos(A)$, the shorthand $L \times R = \sqrt{L} R \sqrt{L}$ is often used, and clearly $L \times R \in \Pos(A)$. For two Hermitian operators $L, K \in \Herm(A)$, we write $L \leq K$ if and only if $K - L \in \cP(A)$ and we write $L \ll K$ if the support of $L$ is contained in the support of $K$.

Let $\cS(A)$ be the set of \emph{quantum states}, i.e., $\cS(A) := \{\rho_A \in \cP(A) \,|\, \tr[\rho_A] = 1 \}$, where $\tr$ denotes the trace. Given a bipartite state $\rho_{AB} \in \cS(AB)$, we denote by $\rho_A = \tr_B[\rho_{AB}]$ its marginal on $A$. We consequently use subscripts to indicate which physical system an operator acts on. Finally, $\pi_A \in \cS(A)$ denotes the maximally mixed state given by $\pi_A = \id_A/|A|$.
%

\subsection{Projectors and Pinching}

We write $\{ L \geq K \} = \id - \{ L < K \}$ for the projector onto the subspace spanned by eigenvectors corresponding to non-negative eigenvalues of $L - K$. By definition we have $(L - K) \{L\geq K \} \geq 0$, and, thus, 
$L \{ L \geq K \} \geq K \{ L \geq K \}$.
For any unitary $V$, we further have
\begin{align}
  V \{ L \geq K \} V^{\dagger} = \big\{ V L V^{\dagger} \geq V K V^{\dagger} \big\} .
  \label{eq:proj-unitary}
\end{align}
We will also use an inequality by Audenaert~\emph{et al.}~\cite[Thm.~1]{audenaert07}, which can be conveniently stated as follows~\cite[Eq.~(24)]{audenaert07-3}. Let $L$ and $K$ be positive semi-definite and $s \in (0,1)$. Then,
\begin{align}
  \tr [L^s K^{1-s}] &\geq \tr \left[ K \{ L \geq K \} \right] + \tr \left[ L \{ L < K \} \right]  . \label{eq:audenaert}
\end{align}

For any Hermitian $L$, we write its spectral decomposition as $L = \sum_{\lambda \in \mathrm{spec}(L)} \lambda P^{L,\lambda}$, where $P^{L,\lambda}$ are projectors and $\mathrm{spec}(L) \subset \mathbb{R}$ is its discrete spectrum.
We denote by $\cP_{L}$ the \emph{pinching map} for this spectral decomposition, i.e.\ the following completely positive trace-preserving map:
\begin{align}
  \cP_{L} : K \mapsto \sum_{\lambda \in \mathrm{spec}(L)} P^{L,\lambda} K P^{L,\lambda} \,. 
\end{align}

\subsection{Permutation Invariance and Universal State}

We will use the following observation from the representation theory of the group~$S_n$ of permutations of $n$ elements. 
Let $U_{A^n}\!: S_n \mapsto \mathcal{U}(A^n)$ denote the natural unitary representation of $S_n$ that permutes the subsystems $A_1, A_2, \ldots, A_n$. An operator $L_{A^n}$ is called \emph{permutation invariant} if it satisfies $U_{A^n}(\pi) L_{A^n} U_{A^n}(\pi)^{\dagger} = L_{A^n}$ for all $\pi \in S_n$. Similarly, we say that $L_{A^n}$ is \emph{invariant under ($n$-fold) product unitaries} if it satisfies $V_A^{\otimes n} L_{A^n} V_A^{\dagger \otimes n} = L_{A^n}$ for all $V_A \in \mathcal{U}(A)$.

\begin{lemma}
  \label{lm:universal}
  Let $A$ be a system with $|A| = d$. For all $n \in \mathbb{N}$ there exists a state $\omega_{A^n}^n \in \cS(A^n)$, which we call \emph{universal state}, such that the following holds: 
  \begin{enumerate}
    \item For all permutation invariant states $\tau_{A^n} \in \cS(A^n)$, we have
  \begin{align}
    \tau_{A^n} \leq g_{n,d}\ \omega_{A^n}^n
  \qquad \textrm{with} \qquad g_{n,d} = { n + d^2 - 1 \choose n } \leq (n + 1)^{d^2-1} \,. \label{eq:universal-state}
  \end{align}
  \item 
     The universal state has the form
     \begin{align}
        \omega_{A^n}^n = \bigoplus_{\lambda \in \Lambda_{n,d}} p_{\lambda} \frac{P_{A^n}^{\lambda}}{\tr(P_{A^n}^{\lambda})}, \label{eq:uni0}
     \end{align}
     where $\Lambda_{n,d}$ is the set of Young diagrams of size $n$ and depth $d$ and satisfies $|\Lambda_{n,d}| \leq (n+1)^{d-1}$, $\{ P_{A^n}^{\lambda} \}_{\lambda}$ are mutually orthogonal projectors and $\{ p_{\lambda} \}_{\lambda}$ is a probability distribution.
  \item The state $\omega_{A^n}^n$ is permutation invariant and invariant under $n$-fold product unitaries, and commutes with all permutation invariant states.
  \end{enumerate}
\end{lemma}

Note that a related construction is presented in~\cite{christandl09}, and we refer the reader to~\cite{matthiasphd,hayashi16} for a thorough discussion of group representation theory in the context of quantum information. A different explicit construction of such a universal state is also proposed in~\cite[Sec.~3]{hayashi09b}, but the constant given there instead of $g_{n,d}$ is not optimal.

\begin{proof}
  Since $\tau_{A^n}$ is invariant under permutations, it has a purification $\tau_{A^nA'^n}$ in the symmetric subspace of
  $(\cH_A \otimes \cH_{A'})^{\otimes n}$ where $\cH_{A'} \equiv \cH_{A}$ are isomorphic (see, e.g.,~\cite[Lem.~4.2.2.]{renner05}). Recall that the symmetric subspace is spanned by all vectors in $(\cH_A \otimes \cH_{A'})^{\otimes n}$ that are invariant under $U_{A^nA^{\prime n}}(\pi)$ for all $\pi \in S_n$.
  Now let $P_{A^nA'^n}^{\,\textrm{symm}}$ denote the projector onto this symmetric subspace, and its dimension by $g_{n,d}$. Then,
  \begin{align}
    \tau_{A^n A'^n} \leq P_{A^nA'^n}^{\,\textrm{symm}}, \qquad \textrm{and} \qquad P_{A^nA'^n}^{\,\textrm{symm}} = \frac{1}{|S_n|} \sum_{\pi \in S_n} U_{A^n}(\pi) \otimes U_{A'^n}(\pi) \,. \label{eq:uni1}
  \end{align}
  Let us now define the universal state as
  \begin{align}
    \omega_{A^n}^n := \frac{1}{g_{n,d}} \tr_{A'^n} \big[ P_{A^nA'^n}^{\,\textrm{symm}} \big] = \frac{1}{g_{n,d}|S_n|} \sum_{\pi \in S_n} \tr[ U_{A'^n}(\pi) ] \, U_{A^n}(\pi) \label{eq:uni2} \,.
  \end{align}
 The state clearly has the desired first property, due to~\eqref{eq:uni1}. Moreover, it is evident from~\eqref{eq:uni2} that $\omega_{A^n}^n$ is invariant under permutations and product unitaries. By the Schur-Weyl duality the natural  representation of $S_n \times \mathcal{U}(A)$ given by $\pi \times V_A \mapsto U_{A^n}(\pi) \cdot V_A^{\otimes n}$ decomposes into different irreducible representations labelled by the Young diagrams in $\Lambda_{n,d}$ and Schur's lemma thus ensures that $\omega_{A^n}^n$ is of the form given in~\eqref{eq:uni0}. 

 The number $|\Lambda_{n,d}|$ is upper bounded by the number of types of strings of length $n$ with $d$ symbols, which in turn is bounded by $(n+1)^{d-1}$. (See, e.g.,~\cite[Eq.~(1)]{hayashi09b}).

 Finally, note that all irreducible representations of $S_n$ of a type $\lambda$ are contained in the support of $P_{A^n}^{\lambda}$. Thus, by Schur's Lemma every permutation invariant state $\tau_{A^n}$ can be written in the block-diagonal form
$\tau_{A^n} = \bigoplus_{\lambda \in \Lambda_{n,d}} \tau_{A^n}^{\lambda}$, 
  where $\tau_{A^n}^{\lambda} \ll P_{A^n}^{\lambda}$. We conclude the proof by noting that states of this form commute with $\omega_{A^n}^n$.
  \end{proof}

\subsection{R\'enyi Divergence}

Let us define the following two families of R\'enyi divergences for $\alpha \in (0, 1) \cup (1, \infty)$.
For any quantum state $\rho \in \cS(A)$ and positive semi-definite operator $\sigma \in \cP(A)$, we define the \emph{R\'enyi relative entropy}~\cite{ohya93} and the \emph{(sandwiched) R\'enyi divergence}~\cite{lennert13,wilde13}, respectively, as follows. If either $\alpha < 1$ and $\rho \not\perp \sigma$ or $\alpha > 1$ and $\rho \ll \sigma$, we define
\begin{align}
  D_{\alpha}(\rho\|\sigma) &:= \frac{1}{\alpha-1} \log \tr \Big[ \sigma^{\frac{1-\alpha}2}\rho^{\alpha} \sigma^{\frac{1-\alpha}2} \Big] \qquad \textrm{and}\\
  \widetilde{D}_{\alpha}(\rho\|\sigma) &:= \frac{1}{\alpha-1} \log \tr \Big[ \big( \sigma^{\frac{1-\alpha}{2\alpha}} \rho \sigma^{\frac{1-\alpha}{2\alpha}}\big)^{\alpha}  \Big] \,,
\end{align}
and else we set $D_{\alpha}(\rho\|\sigma) = \widetilde{D}_{\alpha}(\rho\|\sigma) = \infty$. Here and hereafter we  employ the generalized inverse to take negative powers of positive operators that do not have full support.
Note that both definitions are continuous functions onto the extended positive real axis $\mathbb{R} \cup \{+\infty\}$. (See, e.g.,~\cite[Lem.~13]{lennert13}.) 

The data-processing inequality states that for any completely positive trace preserving map $\mathcal{E}$, we have $D_{\alpha}(\rho\|\sigma) \geq D_{\alpha}(\mathcal{E}(\rho)\|\mathcal{E}(\sigma))$ for $\alpha \in [0,2]$~\cite{petz86} and $\widetilde{D}_{\alpha}(\rho\|\sigma) \geq \widetilde{D}_{\alpha}(\mathcal{E}(\rho)\|\mathcal{E}(\sigma))$ for $\alpha \in [\frac12,\infty)$~\cite{frank13,beigi13,lennert13}.
We also note the following consequence of the data-processing inequality:
\begin{lemma} \label{lm:iso}
  For any states $\rho \in \cS(A)$, $\sigma \in \cS(A')$, and any isometry $U: \cH_A \to \cH_{A'}$, we have
  \begin{align}
    D_{\alpha}(\rho\|U^{\dag} \sigma U) &\leq D_{\alpha}(U\rho U^{\dagger}\| \sigma ) && \textrm{if } \alpha \in [0, 2] \\
    \widetilde{D}_{\alpha}(\rho\|U^{\dag} \sigma U) &\leq \widetilde{D}_{\alpha}(U\rho U^{\dagger}\| \sigma )  && \textrm{if } \alpha \in \Big[\frac12, \infty\Big) \,.
  \end{align}
\end{lemma}

\begin{proof}
  Define the projector onto the image of $U$ as $P_{A'} = U U^{\dag}$. Then we leverage on the data-processing inequality for the map $X \mapsto P_{A'} X P_{A'} + (1_{A'}-P_{A'})X(1_{A'}-P_{A'})$ to find
  \begin{align}
     D_{\alpha}(U\rho U^{\dag}\| \sigma ) &\geq D_{\alpha}\big( P_{A'}U \rho U^{\dag}P_{A'} \big\| P_{A'} \sigma P_{A'} + (1_{A'}-P_{A'})\sigma(1_{A'}-P_{A'}) \big) \\
     &= D_{\alpha}(U\rho U^{\dag} \| U U^{\dag} \sigma U U^{\dag} ) \\
     &= D_{\alpha}( \rho \| U^{\dag} \sigma U)
  \end{align}
  for $\alpha \in [0,2]$. The analogous argument for $\widetilde{D}_{\alpha}$ and $\alpha \in [\frac12, \infty)$ concludes the proof. 
\end{proof}

The two families of divergences coincide when $\rho$ and $\sigma$ commute. For $\alpha \in \{ 0, 1, \infty \}$ we define $D_{\alpha}(\rho\|\sigma)$ and $\widetilde{D}_{\alpha}(\rho\|\sigma)$ as the corresponding limit. The relative entropy emerges when we take the limit $\alpha \to 1$ in both cases, namely~\cite{lennert13,wilde13}
\begin{align}
  \lim_{\alpha \to 1} \widetilde{D}_{\alpha}(\rho\|\sigma) = \lim_{\alpha \to 1} D_{\alpha}(\rho\|\sigma) = \tr \big[ \rho (\log \rho - \log \sigma) \big] =: D(\rho\|\sigma) \,.
\end{align}
Some special cases of these entropies, in particular $D_{0}(\rho\|\sigma)$ and $\widetilde{D}_{\infty}(\rho\|\sigma)$ have previously been discussed in~\cite{datta08} and are based on Renner's min- and max-entropy~\cite{renner05}. A comprehensive overview of other special cases is given in~\cite{lennert13}.

For the second order analysis we will employ the \emph{information variance}~\cite{li12,tomamichel12}, given as
\begin{align}
  V(\rho\|\sigma) := \tr \Big[ \rho \big(\log \rho - \log \sigma - D(\rho\|\sigma) \big)^2 \Big] \,.
\end{align}
In particular, we will use the fact that~\cite[Props.~4--5]{lintomamichel14} (see also~\cite{mybook})
\begin{align}
  \frac{\partial}{\partial \alpha} D_{\alpha}(\rho\|\sigma)  \bigg|_{\alpha = 1} = \frac{\partial}{\partial \alpha} \widetilde{D}_{\alpha}(\rho\|\sigma)  \bigg|_{\alpha = 1} = \frac{1}{2 \log e} V(\rho\|\sigma).
  \label{H-28}
\end{align}

Let $|\spec(\sigma)|$ denote the number of mutually different eigenvalues of $\sigma$. The following property of the sandwiched R\'enyi divergence is crucial for our derivations:
\begin{lemma}
  \label{lm:pinching}
  Let $\rho \in \cS$ and $\sigma \in \cP$. For all $\alpha \geq 0$, we have
  \begin{align}
     \widetilde{D}_{\alpha} (\cP_{\sigma}(\rho) \| \sigma ) \leq \widetilde{D}_{\alpha}(\rho\|\sigma) \leq \widetilde{D}_{\alpha}(\cP_{\sigma}(\rho) \| \sigma ) + \begin{cases} \log |\spec(\sigma)| & \textrm{if } \alpha \in [0, 2] \\
      2 \log |\spec(\sigma)| & \textrm{if } \alpha > 2  \end{cases}\,.
  \end{align}
\end{lemma}

  Since $\cP_{\sigma}(\sigma) = \sigma$, the first inequality is a special case of the data-processing inequality. This special case was first established in~\cite[Prop.~14]{lennert13}.
The second inequality contains Hiai-Petz's evaluation \cite{hiai91} as the special case with $\alpha=1$.
  
\begin{proof}
  It suffices to show the inequalities for the case were all quantities are finite as they otherwise hold trivially due to the fact that $\rho \ll \sigma \iff \cP_{\sigma}(\rho) \ll \sigma$ and $\rho \perp \sigma \iff \cP_{\sigma}(\rho) \perp \sigma$.
  To derive the upper bound for $\alpha \in (1,2]$, we write
  \begin{align}
    \exp \big( (\alpha - 1) \widetilde{D}_{\alpha}(\rho\|\sigma) \big) 
    &= \tr \big[ \big( \sigma^{\frac{1-\alpha}{2\alpha}} \rho \sigma^{\frac{1-\alpha}{2\alpha}} \big)^{\alpha-1}  \sigma^{\frac{1-\alpha}{2\alpha}} \rho \sigma^{\frac{1-\alpha}{2\alpha}} \big] \\
    &\leq \tr \big[ \big( \sigma^{\frac{1-\alpha}{2\alpha}} |\spec(\sigma)| \cP_{\sigma}(\rho) \sigma^{\frac{1-\alpha}{2\alpha}} \big)^{\alpha-1} \sigma^{\frac{1-\alpha}{2\alpha}} \rho \sigma^{\frac{1-\alpha}{2\alpha}} \big] \label{eq:pinching-new-ineq}\\
    &= |\spec(\sigma)|^{\alpha -1} \tr \big[ \big( \sigma^{\frac{1-\alpha}{2\alpha}} \cP_{\sigma}(\rho) \sigma^{\frac{1-\alpha}{2\alpha}} \big)^{\alpha-1} \sigma^{\frac{1-\alpha}{2\alpha}} \cP_{\sigma}(\rho) \sigma^{\frac{1-\alpha}{2\alpha}} \big] \\
    &= |\spec(\sigma)|^{\alpha -1} \exp \big( (\alpha - 1) \widetilde{D}_{\alpha}(\cP_{\sigma}(\rho) \|\sigma\big) \big) \label{eq:pinching-new-final}.
  \end{align}
  To establish~\eqref{eq:pinching-new-ineq}, we use~\cite[Lem.~9]{hayashi06} which states that
  \begin{align}
    \rho \leq |\spec(\sigma)|\, \cP_{\sigma}(\rho) \label{eq:pinching}
  \end{align}
  and, since the function $x \mapsto x^{\alpha-1}$ is operator monotone for $\alpha \in (1, 2)$,
  \begin{align}
    \big( \sigma^{\frac{1-\alpha}{2\alpha}} \rho \sigma^{\frac{1-\alpha}{2\alpha}} \big)^{\alpha-1}\leq |\spec(\sigma)|^{\alpha-1} \, \big( \sigma^{\frac{1-\alpha}{2\alpha}} \cP_{\sigma}(\rho) \sigma^{\frac{1-\alpha}{2\alpha}} \big)^{\alpha-1}
    \label{eq:pinching-ineq-applied} .
  \end{align}
  
   An analogous argument, with the opposite inequality 
   \eqref{eq:pinching-new-ineq}, holds for $\alpha \in (0, 1)$.
  Thus, for all $\alpha \in (0, 1) \cup (1, 2]$, we conclude 
  that
  \begin{align}
    \widetilde{D}_{\alpha} (\rho\|\sigma) \leq \widetilde{D}_{\alpha}(\cP_{\sigma}(\rho) \|\sigma\big) + \log |\spec(\sigma)| . \label{eq:stronger-upper}
  \end{align}
  
  To get an upper bound for $\alpha > 2$ we observe that
  \begin{align}
        \exp \big( (\alpha - 1) \widetilde{D}_{\alpha}(\rho\|\sigma) \big) &= \tr \big[ \big( \sigma^{\frac{1-\alpha}{2\alpha}} \rho \sigma^{\frac{1-\alpha}{2\alpha}} \big)^{\alpha} \big] \\
        &\leq |\spec(\sigma)|^{\alpha} \, \tr \big[ \big( \sigma^{\frac{1-\alpha}{2\alpha}} \cP_{\sigma}(\rho) \sigma^{\frac{1-\alpha}{2\alpha}} \big)^{\alpha} \big]
  \end{align}
  since $A \leq B$ implies $\tr[f(A)] \leq \tr[f(B)]$ for every monotonically increasing function $f$. Thus,
  \begin{align}
    \widetilde{D}_{\alpha} (\rho\|\sigma) \leq \widetilde{D}_{\alpha}(\cP_{\sigma}(\rho) \|\sigma\big) + \frac{\alpha}{\alpha-1} \log |\spec(\sigma)| \leq \widetilde{D}_{\alpha}(\cP_{\sigma}(\rho) \|\sigma\big) + 2 \log |\spec(\sigma)| .
  \end{align}
  Finally, note that the inequality thus also holds for the limiting cases $\alpha \in \{0, 1, \infty\}$. 
\end{proof}

    Note that if $\sigma$ has $k$ different eigenvalues, the number of different eigenvalues of $\sigma^{\otimes n}$ is upper bounded by ${n+k-1 \choose k-1} \leq (n+1)^{d-1}$ (see, e.g.,~\cite{csiszar98}).
As a direct consequence of this and Lemma~\ref{lm:pinching}, we find
\begin{corollary}
  \label{cor:achieve}
  Let $\rho \in \cS$ and $\sigma \in \cP$. For all $\alpha \geq 0$, we have
  \begin{align}
    \lim_{n \to \infty} \left\{ \frac{1}{n} D_{\alpha}\big( \cP_{\sigma^{\otimes n}} (\rho^{\otimes n}) \,\big\|\, \sigma^{\otimes n} \big)  \right\}
    = \widetilde{D}_{\alpha}(\rho\|\sigma) \,. \label{eq:measurement-achievability}
  \end{align}
\end{corollary}
This extends a prior result by Mosonyi and Ogawa in~\cite{mosonyiogawa13} to all $\alpha \geq 0$,
and includes Hiai-Petz's result \cite{hiai91} as the case with $\alpha=1$.



\section{Generalized R\'enyi Mutual Information}
\label{sc:mui}

We state our results in a general form that allows us to treat mutual information and conditional entropies at the same time. The more common definitions of the R\'enyi mutual information and the R\'enyi conditional entropies can then be recovered as special cases.

\subsection{Definitions and Basic Properties}

The mutual information can be expressed in different forms. Let $\rho_{AB} \in \cS(AB)$ be a bipartite state, then the \emph{mutual information} is given as
\begin{align}
  I(A\!:\!B)_{\rho} := D(\rho_{AB}\| \rho_A \otimes \rho_B) \label{eq:mi1} 
  &= \min_{\sigma_B \in \cS(B)} D(\rho_{AB}\|\rho_A \otimes \sigma_B) \\
  &= \min_{\tau_A \in \cS(A),\atop \sigma_B \in \cS(B)} D(\rho_{AB}\|\tau_A \otimes \sigma_B) \label{eq:mi2}.
\end{align}
The latter two equalities essentially follow by the chain rule $D(\rho_{AB}\|\tau_A\otimes\sigma_B) = D(\rho_{AB}\|\rho_A\otimes\rho_B) + D(\rho_A\|\tau_A) + D(\rho_B\|\sigma_B)$ and the positive definiteness of the relative entropy, which in particular implies that
the minimizers are given by the marginals $\rho_A$ and $\rho_B$.

Clearly, this allows for various, potentially different extensions to $\alpha \neq 1$. In this work we mostly focus on the expression on the right hand side of~\eqref{eq:mi1}. 
Let us therefore define the \emph{generalized R\'enyi mutual information} and the \emph{generalized sandwiched R\'enyi mutual information} for a state $\rho_{AB} \in \cS(AB)$ and any $\tau_A \geq 0$ such that $\tau_A \gg \rho_A$ as follows: 
\begin{align}
  \label{eq:ialpha}
  I_{\alpha}(\rho_{AB} \| \tau_A)  &:= \inf_{\sigma_B \in \cS(B)} D_{\alpha}(\rho_{AB}\|\tau_A \otimes \sigma_B) \,,\\
  \widetilde{I}_{\alpha}(\rho_{AB} \| \tau_A) &:=  \inf_{\sigma_B \in \cS(B)} \widetilde{D}_{\alpha}(\rho_{AB}\|\tau_A \otimes \sigma_B)\,. \label{eq:ialpha2}
\end{align}
It is easy to verify that the infimum in the above definitions can be replaced by a minimum since the divergences are lower semi-continuous in $\sigma_B$ (as can verified immediately from the definition in the finite-dimensional case). Moreover, without loss of generality, we can always restrict the spaces $\cH_A$ and $\cH_B$ to the supports of $\rho_A$ and $\rho_B$, respectively, as a consequence of Lemma~\ref{lm:iso} when $\alpha$ satisfies the constraints given there.
%
For later reference, we also define
\begin{align}
  V(\rho_{AB} \|\tau_A) := V(\rho_{AB}\|\tau_A \otimes \rho_B) .
\end{align}

Note that the R\'enyi mutual information and the sandwiched R\'enyi mutual information~\cite{beigi13,gupta13} is 
recovered by choosing $\tau_A = \rho_A$, namely we define 
\begin{align}
  I_{\alpha}^{\downarrow}(A\!:\!B)_{\rho} 
:= I_{\alpha}(\rho_{AB}\|\rho_A)\,, \qquad \textrm{and} \qquad
  \widetilde{I}_{\alpha}^{\downarrow}(A\!:\!B)_{\rho} 
:= \widetilde{I}_{\alpha}(\rho_{AB}\|\rho_A) \,.
\end{align}

%
Similarly, we define the R\'enyi conditional entropy~\cite{tomamichel13} and the sandwiched R\'enyi conditional entropy~\cite{lennert13} by choosing $\tau_A = \id_A$. Using the notation of~\cite{tomamichel13}, we have
\begin{align}
  H_{\alpha}^{\uparrow}(A|B)_{\rho} &:= -I_{\alpha}(\rho_{AB}\|\id_A) = \log |A| - I_{\alpha}(\rho_{AB}\|\pi_A) \qquad \textrm{and} \\
    \widetilde{H}_{\alpha}^{\uparrow}(A|B)_{\rho} &:= -\widetilde{I}_{\alpha}(\rho_{AB}\|\id_A) = \log |A| - \widetilde{I}_{\alpha}(\rho_{AB}\|\pi_A) \,.
\end{align}

\subsection{Characterization of the Minimizers}

Specializing the definitions in~\eqref{eq:ialpha} and~\eqref{eq:ialpha2} to $\alpha = 1$, we find that
\begin{align}
  I(\rho_{AB} \| \tau_A) := \min_{\sigma_B \in \cS(B)} D(\rho_{AB}\|\tau_A \otimes \sigma_B) = D(\rho_{AB}\| \tau_A \otimes \rho_B) \,,
\end{align}
i.e.\ the minimizer is given by the marginal $\rho_B$ at $\alpha = 1$ as noted above. For $I_{\alpha}(\rho_{AB}\|\tau_A)$ and general $\alpha$ we can determine the minimizer using a quantum Sibson's identity (see~\cite[Lem.~3 in Suppl.~Mat.]{sharma13}), which yields the unique minimizer of~\eqref{eq:ialpha} for all $\alpha$. Namely,
   \begin{equation}
   I_{\alpha}(\rho_{AB}\|\tau_A) = D_{\alpha}(\rho_{AB}\|\tau_A \otimes \sigma_B^*(\alpha)) 
   \quad \textrm{with} \quad 
   \sigma_B^*(\alpha) := \frac{ \left( \tr_A \left[ \tau_A^{{1-\alpha}}  \rho_{AB}^{\alpha} \right] \right)^{\frac{1}{\alpha}}}{\tr\Big[\left( \tr_A \left[ \tau_A^{{1-\alpha}}  \rho_{AB}^{\alpha} \right] \right)^{\frac{1}{\alpha}} \Big]} \,. \label{eq:sibson-def}
   \end{equation}

Such a characterization is harder to find for the minimizer in~\eqref{eq:ialpha2}. 
The following lemma, proved in Appendix~\ref{app:diff}, 
provides a characterization of the minimizer for general $\alpha$ as a fixed-point of a peculiar non-linear map.

\begin{lemma}
\label{pr:fixed-point-lemma}
  Let $\rho_{AB} \in \cS(AB)$ and $\tau_A \geq 0$ such that $\tau_A \gg \rho_A$ and $\alpha \geq \frac12$. Then there exists a unique state $\widetilde{\sigma}_B^*(\alpha) \gg \rho_B$ such that
 \begin{align}
    \widetilde{\sigma}_B^*(\alpha) = \frac{\tr_A \Big[ \Big( \big(\tau_A \otimes \widetilde{\sigma}_B^*(\alpha)\big)^{\frac{1-\alpha}{2\alpha}} \rho_{AB} \big(\tau_A \otimes \widetilde{\sigma}_B^*(\alpha)\big)^{\frac{1-\alpha}{2\alpha}} \Big)^{\alpha} \Big] }{ \tr \Big[ \Big( \big(\tau_A \otimes \widetilde{\sigma}_B^*(\alpha)\big)^{\frac{1-\alpha}{2\alpha}} \rho_{AB} \big(\tau_A \otimes \widetilde{\sigma}_B^*(\alpha)\big)^{\frac{1-\alpha}{2\alpha}} \Big)^{\alpha} \Big]} \,.
  \end{align}
  which satisfies $\widetilde{I}_{\alpha}(\rho_{AB}\|\tau_A) = \widetilde{D}_{\alpha}(\rho_{AB}\|\tau_A \otimes \widetilde{\sigma}_B^*(\alpha))$
 \end{lemma}
For $\alpha = 1$ the fixed-point of this map is clearly given uniquely by $\widetilde{\sigma}_B^*(1) = \rho_B$, as expected.

\subsection{Duality Relations}

We will take advantage of the following duality relation for the mutual information:
\begin{lemma}
  \label{lm:duality}
  Let $\rho_{AB} \in \cS(AB)$, $\tau_A \geq 0$ such that $\tau_A \gg \rho_A$. Then, for any purification $\rho_{ABC}$ of $\rho_{AB}$, we have
\begin{align}
   \widetilde{I}_{\alpha}(\rho_{AB}\|\tau_A) &= -\widetilde{I}_{\beta}(\rho_{AC}\|\tau_A^{-1}) \,, \quad &&\textrm{for} \quad  \alpha, \beta \in \big[1/2,\infty\big), \quad \alpha^{-1} + \beta^{-1} = 2  \label{eq:duality1}, \\
   I_{\alpha}(\rho_{AB}\|\tau_A) &= - \widetilde{D}_{\beta}(\rho_{AC} \| \tau_A^{-1} \otimes \rho_C) 
\quad &&\textrm{for} \quad \alpha \in [0,\infty), \quad \beta = \alpha^{-1} \label{eq:duality2} , \\
D_\alpha(\rho_{AB}\|\tau_A\otimes \rho_B)
&=-D_\beta(\rho_{AC}\|\tau_A^{-1}\otimes \rho_C)
\quad &&\textrm{for} \quad 
\alpha,\beta \in [0,2], \quad  \alpha+\beta=2
\label{eq:duality3},
\end{align}
where the inverse is taken on the support of $\tau_A$.
\end{lemma}
The first equation is a rather straightforward generalization of the duality relation for the conditional R\'enyi entropy that was recently established independently in~\cite[Thm.~10]{lennert13} and in~\cite{beigi13}. The second equation similarly generalizes the duality relation in~\cite[Thm.~2]{tomamichel13}.
We provide a proof in Appendix~\ref{app:dual} for completeness.

Finally, combining \eqref{H-28} and \eqref{eq:duality3}, we find
\begin{align}
V(\rho_{AC}\|\tau_A^{-1})=
V(\rho_{AB}\|\tau_A)
\label{eq:duality4}.
\end{align}

\subsection{Additivity}

We are interested in the additivity of the mutual informations $I_{\alpha}$ and $\widetilde{I}_{\alpha}$, which follow immediately from the duality relations established in the previous section. To see this, note that the inequality trivially holds in one direction by definition\,---\,the other direction then follows by applying the duality relation on both sides for a product purification. (We note that a simple extension of~\cite[Th.~10 and 11]{beigi13} also establishes that $\widetilde{I}_{\alpha}(\rho_{AB}\|\tau_A)$ is additive.)

\begin{lemma}
  \label{lm:add-old}  \label{lm:add-new}
  Let $\rho_{AB} \in \cS(AB)$, $\omega_{A'B'} \in \cS(A'B')$, and $\tau_{A}, \kappa_{A'} \geq 0$ with $\tau_A \gg \rho_A$, $\kappa_{A'} \gg \omega_{A'}$. Then, 
  \begin{align}
    I_{\alpha}\big(\rho_{AB} \otimes \omega_{A'B'} \| \tau_{A} \otimes \kappa_{A'} \big) &= I_{\alpha}(\rho_{AB}\|\tau_A) + I_{\alpha}(\omega_{A'B'}\|\kappa_{A'}), \quad && \textrm{for} \quad \alpha \geq 0, \qquad \textrm{and}\\ 
  \widetilde{I}_{\alpha}\big(\rho_{AB} \otimes \omega_{A'B'} \| \tau_{A} \otimes \kappa_{A'} \big) &= \widetilde{I}_{\alpha}(\rho_{AB}\|\tau_A) + \widetilde{I}_{\alpha}(\omega_{A'B'}\|\kappa_{A'}), \quad && \textrm{for} \quad \alpha \geq \frac12 . 
  \end{align}
\end{lemma}

So, in particular, we find that the mutual information quantities $I_{\alpha}^{\downarrow}(A\!:\!B)_{\rho}$ and $\widetilde{I}_{\alpha}^{\downarrow}(A\!:\!B)_{\rho}$ as well as the conditional entropies $H_{\alpha}^{\uparrow}(A|B)$ and $\widetilde{H}_{\alpha}^{\uparrow}(A|B)$ are additive for product states (for the ranges of $\alpha$ indicated in the lemma).

\subsection{Uniform Asymptotic Achievability}


The following result forms the core of our proof for the achievability of the strong converse exponent. It establishes that the mutual information $\widetilde{I}_{\alpha}(\rho_{AB}\|\tau_A)$ can be expressed as a limit of classical R\'enyi divergences. More precisely, we have the following.

\begin{proposition}
  \label{pr:universal}
  Let $\rho_{AB} \in \cS(AB)$ and $\tau_A \in \cS(A)$ such that $\tau_A \gg \rho_A$. For any $\alpha \geq \frac{1}{2}$, we have
  \begin{align}
    \label{eq:universal}
   \frac{1}{n} D_{\alpha} \big(\cP_{\tau_A^{\otimes n} \otimes\, \omega_{B^n}^n} ( \rho_{AB}^{\otimes n} ) \big\|\, \tau_A^{\otimes n} \otimes \omega_{B^n}^n \big) = \widetilde{I}_{\alpha}(\rho_{AB}\|\tau_A) + O\left(\frac{\log n}{n}\right) \,,
  \end{align}
  where $\omega_{B^n}^n$ is the universal state of Lemma~\ref{lm:universal}.
  Moreover, the implied constants are independent of $\alpha$.
\end{proposition}

\begin{proof}
  For any $\sigma_B \in \cS(B)$, employing the data-processing inequality we find
\begin{align}
  \widetilde{D}_{\alpha} \big(\cP_{\tau_A^{\otimes n} \otimes\, \omega_{B^n}^n} ( \rho_{AB}^{\otimes n} ) \,\big\|\, \tau_A^{\otimes n} \otimes \omega_{B^n}^n \big) &\leq \widetilde{D}_{\alpha} \big( \rho_{AB}^{\otimes n} \,\big\|\, \tau_A^{\otimes n} \otimes \omega_{B^n}^n \big) \\
  &\leq \widetilde{D}_{\alpha} \big( \rho_{AB}^{\otimes n} \,\big\|\, \tau_A^{\otimes n} \otimes \sigma_B^{\otimes n} \big) + \log g_{n,d} ,
\end{align}
where we used~\cite[Prop.~5]{lennert13} in the last step and set $d = \max\{|A|,|B|\}$ for later convenience.
Then,
\begin{align}
    \frac{1}{n} \widetilde{D}_{\alpha} \big(\cP_{\tau_A^{\otimes n} \otimes\, \omega_{B^n}^n} ( \rho_{AB}^{\otimes n} ) \,\big\|\, \tau_A^{\otimes n} \otimes \omega_{B^n}^n \big) &\leq \min_{\sigma_B \in \cS(B)} \widetilde{D}_{\alpha} \big( \rho_{AB} \,\big\|\, \tau_A \otimes \sigma_B \big) + \frac{\log g_{n,d}}{n} \\
    &= \widetilde{I}_{\alpha}(\rho_{AB}\|\tau_A) + \frac{\log g_{n,d}}{n}  .
\end{align}

The upper bound then follows from the fact that $g_{n,d}$ grows polynomially in $n$. For the lower bound, we first invoke Lemma~\ref{lm:pinching}, which yields
\begin{align}
  \widetilde{D}_{\alpha} \big(\cP_{\tau_A^{\otimes n} \otimes\, \omega_{B^n}^n} ( \rho_{AB}^{\otimes n} ) \,\big\|\, \tau_A^{\otimes n} \otimes \omega_{B^n}^n \big) &\geq \widetilde{D}_{\alpha}\big(\rho_{AB}^{\otimes n}\,\big\|\, \tau_A^{\otimes n} \otimes \omega_{B^n}^n \big) - \log 
  \big( v_{\tau_A^{\otimes n} \otimes\, \omega_{B^n}^n} \big) \\
  &\geq \widetilde{I}_{\alpha}\big(\rho_{AB}^{\otimes n} \,\big\|\, \tau_{A}^{\otimes n} \big) - \log \big(v_{\tau_A^{\otimes n} \otimes\, \omega_{B^n}^n} \big) .
\end{align}
where we use the shorthand notation $v_{\sigma} = |\spec(\sigma)|$.
Next, we recall that Lemma~\ref{lm:add-new} establishes the additivity of $\widetilde{I}_{\alpha}$, in particular $\widetilde{I}_{\alpha}(\rho_{AB}^{\otimes n} \,\|\, \tau_{A}^{\otimes n} ) = n \widetilde{I}_{\alpha}(\rho_{AB}\|\tau_A)$. Thus, we have
\begin{align}
  \frac{1}{n} \widetilde{D}_{\alpha} \big(\cP_{\tau_A^{\otimes n} \otimes\, \omega_{B^n}^n} ( \rho_{AB}^{\otimes n} ) \,\big\|\, \tau_A^{\otimes n} \otimes \omega_{B^n}^n \big) &\geq \widetilde{I}_{\alpha}(\rho_{AB}\|\tau_A) + \frac{\log \big(v_{\tau_A^{\otimes n} \otimes\, \omega_{B^n}^n}\big)}{n} .
\end{align}
Finally, note that 
  $v_{\tau_A^{\otimes n} \otimes\, \omega_{B^n}^n} \leq v_{\tau_A^{\otimes n}}  v_{\omega_{B^n}} \leq (n+1)^{2 (d - 1)}$
 to conclude the proof.
\end{proof}

It is important that the correction terms in the above derivation are of the order $o\big(n^{-\frac12}\big)$. This allows for the following corollary.

\begin{corollary}
  \label{corr:2nd-order}
  For any $t \in \mathbb{R}$, we have
  \begin{align}
  \lim_{n \to \infty} \left\{ \frac{t}{\sqrt{n}} \bigg( D_{1+\frac{t}{\sqrt{n}}} \Big(\cP_{\tau_A^{\otimes n} \otimes\, \omega_{B^n}^n} ( \rho_{AB}^{\otimes n} ) \,\Big\|\, \tau_A^{\otimes n} \otimes \omega_{B^n}^n \Big) - n I(\rho_{AB}\|\tau_A) \bigg) \right\} = \frac{t^2}{2 \log e} V(\rho_{AB}\| \tau_A )  \,.
  \end{align}
\end{corollary}

\begin{proof}
   According to Proposition~\ref{pr:diff} below, the Taylor expansion of $\widetilde{I}_{\alpha}(\rho_{AB}\|\tau_A)$ for $\alpha$ close to $1$ is
   \begin{align}
     \widetilde{I}_{\alpha}(\rho_{AB}\|\tau_A) = D(\rho_{AB}\|\tau_A \otimes \rho_B) 
     + \frac{(\alpha-1)}{2 \log e} V(\rho_{AB}\|\tau_A \otimes \rho_B) + o\left( (\alpha-1) \right) .
   \end{align}
   Substituting this into~\eqref{eq:universal} with $\alpha = 1 + \frac{t}{\sqrt{n}}$ yields the desired limit.
\end{proof}

In particular this establishes that the function $\alpha \mapsto \widetilde{I}_{\alpha}(\rho_{AB}\|\tau_A)$ is a pointwise limit of a sequence of classical R\'enyi divergences and the convergence is uniform in $\alpha$. From this we deduce that the function inherits continuity and monotonicity from the classical analgue.
\begin{corollary}
  \label{corr:point-wise}
  The function $\alpha \mapsto \widetilde{I}_{\alpha}(\rho_{AB}\|\tau_A)$ is continuous and monotonically increasing. Moreover, the function $t \mapsto t \widetilde{I}_{1+t}(\rho_{AB}\|\tau_A)$ is continuous and convex.
\end{corollary}

Clearly these results immediately specialize to $\widetilde{I}_{\alpha}^{\downarrow}(A\!:\!B)_{\rho}$ and $\widetilde{H}_{\alpha}(A|B)_{\rho}$.

\subsection{Differentiability in $\alpha$}

The argument used to derive Corollary~\ref{corr:point-wise} does not suffice to establish differentiability of the sandwiched R\'enyi mutual information.

\begin{proposition}
  \label{pr:diff}
  Let $\rho_{AB} \in \cS(AB)$ and $\tau_A \in \cS(A)$ such that $\tau_A \gg \rho_A$. Then, the function $\alpha \mapsto \widetilde{I}_{\alpha}(\rho_{AB}\|\tau_A)$ is continuously differentiable for $\alpha \geq \frac12$ with
  \begin{align}
      \frac{\mathrm{d}}{\mathrm{d} \alpha} \widetilde{I}_{\alpha}(\rho_{AB}\|\tau_A) = \frac{\partial}{\partial \alpha} \widetilde{D}_{\alpha}(\rho_{AB}\|\tau_A \otimes \sigma_B) \Big|_{\sigma_B = \widetilde{\sigma}_B^*(\alpha)} ,
  \end{align}
  where $\widetilde{\sigma}_B^*(\alpha)$ is the optimizer in Lemma~\ref{pr:fixed-point-lemma}.
  In particular,
 $\frac{\mathrm{d} }{\mathrm{d}  \alpha} \widetilde{I}_{\alpha}(\rho_{AB}\|\tau_A) \big|_{\alpha = 1} = \frac{1}{2 \log (e)} V(\rho_{AB} \| \tau_A)$.
\end{proposition}

Let us remark that continuity of $\alpha \mapsto \widetilde{I}_{\alpha}(\rho_{AB}\|\tau_A)$ also follows  from the fact that $\alpha \mapsto \widetilde{D}_{\alpha}(\rho_{AB}\|\tau_A \otimes \sigma_B)$ is continuous and the duality relation in Lemma~\ref{lm:duality}. However, due to the optimization over $\sigma_B$ involved in the definition of $\widetilde{I}_{\alpha}(\rho_{AB}\|\tau_A)$, it is not at all clear that the function is differentiable. We show this proposition in Appendices~\ref{app:diff} and~\ref{app:diff2}.

\section{Problem Definition and Operational Quantities}
\label{sc:prob}

We define a more general hypothesis testing problem that allows us to treat both problems discussed in the introduction together.

Let $\rho_{AB} \in \cS(AB)$ be a bipartite quantum state on systems $A$ and $B$ and let $\tau_A \in \cS(A)$ be a state on system $A$. Throughout this paper we assume that $\rho_A \ll \tau_A$. We are interested in the following \emph{composite hypothesis testing} problem:
\begin{align}
  \textrm{null hypothesis:} & \qquad \textrm{state is}\  \rho_{AB} \\
  \textrm{alternative hypothesis:} & \qquad \textrm{state is}\ \tau_A \otimes \sigma_B, \textrm{ for some state } 
  \sigma_B \in \cS(B).
\end{align}

We consider arbitrary bipartite hypothesis tests, given by an operator $0 \leq Q_{AB} \leq \id$ on $AB$ and define the \emph{type-I error} and \emph{type-II error}, respectively, as follows:
\begin{align}
  \alpha(Q_{AB}; \rho_{AB}) &:= \tr\big[ ( \id - Q_{AB}) \rho_{AB} \big], \quad \qquad \textrm{and}\\
  \beta(Q_{AB}; \tau_A) &:= \max_{\sigma_B \in \cS(B)} 
  \tr\big\{ Q_{AB} (\tau_A \otimes \sigma_B) \big\} .
\end{align}
It is convenient to define the quantity $\hat{\alpha}(\mu; \rho_{AB}, \tau_A)$ as the minimum type-I error when the type-II error is below $\mu$, i.e.\ we consider the following optimization problem:
\begin{align}
  \hat{\alpha}(\mu; \rho_{AB} \| \tau_A) := 
\min_{0 \leq Q_{AB} \leq \id} \Big\{ \alpha(Q_{AB}; \rho_{AB}) \,\Big|\, \beta(Q_{AB}; \tau_A) \leq \mu \Big\} 
\end{align}
and note that this quantity can trivially be bounded as
\begin{align}
  \hat{\alpha}(\mu; \rho_{AB} \| \tau_A) &\geq \max_{\sigma_B \in \cS(B)} \min_{0 \leq Q_{AB} \leq \id \atop \tr [Q_{AB} (\tau_A \otimes \sigma_B)] \leq \mu}  \tr \big[ (\id_{AB} - Q_{AB}) \rho_{AB} \big]  \\
  &\geq  \min_{0 \leq Q_{AB} \leq \id \atop \tr [Q_{AB} (\tau_A \otimes \tau_B)] \leq \mu}  \tr \big[ (\id_{AB}-Q_{AB}) \rho_{AB} \big] \\
  &=: \hat{\alpha}(\mu; \rho_{AB} \| \tau_A \otimes \tau_B) \,,
   \label{eq:hatbound2}
\end{align}
for any $\tau_B \in \cS(B)$. The quantity $\hat{\alpha}(\mu; \rho_{AB} \| \tau_A \otimes \tau_B)$ describes the corresponding minimal error probability for binary hypothesis testing between $\rho_{AB}$ and $\tau_A \otimes \tau_B$, which is well understood.
On the other hand, the quantity $\hat{\alpha}(\mu; \rho_{AB}\|\tau_A)$ is the object of our study here. More precisely,
for any fixed $n \in \mathbb{N}$, we consider the  
following $n$-fold extension of this composite hypothesis testing problem:
\begin{align}
  \textrm{null hypothesis:} & \qquad \textrm{state is}\  \rho_{AB}^{\otimes n} \\
  \textrm{alternative hypothesis:} & \qquad \textrm{state is}\ \tau_A^{\otimes n} \otimes \sigma_{B^n}, \textrm{ for some state } \sigma_{B^n} \in \cS(B^n).
\end{align}
Here, it is important to note that $\sigma_{B^n}$ is an arbitrary state in $\cS(B^n)$, and not restricted to product or permutation invariant states.
We are interested in the asymptotic behavior of $\hat{\alpha}\big(\mu_n; \rho_{AB}^{\otimes n} \big\| \tau_{A}^{\otimes n}\big)$ for suitably chosen sequences $\{\mu_n \}_n$ for large $n$.


\section{Hoeffding Bound}
\label{sc:hoeffding}

Our first result considers the case where the error of the second kind goes to zero exponentially with a rate below the mutual information $I(\rho_{AB}\|\tau_A)$. In this case, we find that the error of the first kind converges to zero exponentially fast, and the exponent is determined by the generalized R\'enyi mutual information, $I_{\alpha}(\rho_{AB}\|\tau_A)$, for $\alpha < 1$.

\begin{theorem}\label{th:hoeffding}
  Let $\rho_{AB} \in \cS(AB)$ and $\tau_{A} \in \cS(A)$. Then, for any $R > 0$, we have
  \begin{equation}
  \label{eq:hoeffding-thm}
    \lim_{n \to \infty} \left\{ - \frac{1}{n} \log \hat{\alpha}\Big(\!\exp(-n R); \rho_{AB}^{\otimes n} \Big\| \tau_A^{\otimes n}\Big) \right\} = \sup_{s \in (0, 1)} \left\{ \frac{1-s}{s} \big( I_{s}(\rho_{AB}\| \tau_A) - R \big) \right\}.
  \end{equation}
\end{theorem}
Note that if $R \geq I(\rho_{AB}\|\tau_A)$ the right hand side of~\eqref{eq:hoeffding-thm} evaluates to zero, revealing that in this case the error of the first kind will decay slower than exponential in $n$. (In fact, we will see in Theorems~\ref{th:han} and~\ref{th:second}, respectively, that the error of the first kind will converge to $1$ exponentially fast in $n$ if $R > I(\rho_{AB}\|\tau_A)$, and that it will converge to $\frac12$ if $R = I(\rho_{AB}\|\tau_A)$.)
Furthermore, if $R < I_0(\rho_{AB}\|\tau_A)$, we find that the right hand side of~\eqref{eq:hoeffding-thm} diverges to $+\infty$ indicating that the decay is faster than exponential in $n$. This includes the case where the error of the first kind is identically zero for sufficiently large $n$, e.g.\ in zero-error channel coding.

We also consider the following two special cases:
\begin{corollary}
    Let $\rho_{AB} \in \cS(AB)$. Then, for any $R > 0$, we have
  \begin{align}
    \lim_{n \to \infty} \left\{ - \frac{1}{n} \log \hat{\alpha}\Big(\!\exp(-n R); \rho_{AB}^{\otimes n} \Big\| \rho_A^{\otimes n}\Big) \right\} &= \sup_{s \in (0, 1)} \left\{ \frac{1-s}{s} \big( I_{s}(A\!:\!B^{\downarrow})_{\rho} - R \big) \right\}, \\
        \lim_{n \to \infty} \left\{ - \frac{1}{n} \log \hat{\alpha}\Big(\!\exp(-n R); \rho_{AB}^{\otimes n} \Big\| \pi_A^{\otimes n}\Big) \right\} &= \sup_{s \in (0, 1)} \left\{ \frac{1-s}{s} \big( \log |A| - H_{s}^{\uparrow}(A|B)_{\rho} - R \big) \right\} .
  \end{align}
\end{corollary}
This corollary establishes an operational interpretation of the R\'enyi mutual information, $I_{\alpha}^{\downarrow}(A:B)_{\rho}$, as well as the R\'enyi conditional entropies, $H_{\alpha}^{\uparrow}(A|B)_{\rho}$, for $0 \leq \alpha \leq 1$.

In the following, we treat the proof of the achievability and optimality in Theorem~\ref{th:hoeffding} separately.

\subsection{Proof of Achievability}

The achievability is shown using a quantum Neyman-Pearson test comparing $\rho_{AB}^{\otimes n}$ with $\tau_A^{\otimes n} \otimes \omega_{B^n}^n$, where $\omega_{B^n}^n$ is the universal state defined in Lemma~\ref{lm:universal}. The analysis follows the lines of the proof of the direct part of the quantum Hoeffding bound given in~\cite[Sec.~5.5]{audenaert07-3} and further hinges on the additivity of the mutual information expressed in Lemma~\ref{lm:add-old}.

In the following we show that, for any $0 < R < I(\rho_{AB}\|\tau_A)$,
\begin{align}
\liminf_{n \to \infty} \left\{ - \frac{1}{n} \log \hat{\alpha}\Big(\!\exp(-n R); \rho_{AB}^{\otimes n} \Big\| \tau_A^{\otimes n}\Big) \right\} \geq \sup_{s \in (0, 1)} \left\{ \frac{1-s}{s} \big( I_{s}(\rho_{AB}\| \tau_A) - R \big) \right\}. \label{eq:hoeffding-thm2}
\end{align}
Note that the expression on the right hand side of~\eqref{eq:hoeffding-thm2} is zero if $R \geq I(\rho_{AB}\|\tau_A)$ and the inequality thus holds trivially for that case. 

\begin{proof}[Proof of Achievability in Theorem~\ref{th:hoeffding}]
  Let us fix any $s \in (0,1)$ for the moment. Moreover, let $\{\lambda_n \}_{n \in \mathbb{N}}$ be real numbers to be specified later and define the sequence of tests
  \begin{align}
    Q_{A^nB^n}^n := \left\{ \rho_{AB}^{\otimes n} \geq \exp(\lambda_n) \, \tau_{A}^{\otimes n} \otimes \omega_{B^n}^n \right\} ,
  \end{align}
  where  $\omega_{B^n}^n$ is the universal state introduced in Lemma~\ref{lm:universal}. First, note that the natural representation of $S_n$ decomposes as $U_{A^nB^n}(\pi) = U_{A^n}(\pi) \otimes U_{B^n}(\pi)$ and that $Q_{A^nB^n}^n$ is permutation invariant 
as a direct consequence of Eq.~\eqref{eq:proj-unitary}. Thus, we have
  \begin{align}
    \beta(Q_{A^nB^n}^n; \tau_A^{\otimes n}) &= \max_{\sigma_{B^n} \in \cS(B^n)} \tr \left[ Q_{A^nB^n}^n\, (\tau_A^{\otimes n} \otimes \sigma_{B^n}) \right] \label{eq:ah-first} \\
    &= \max_{\sigma_{B^n} \in \cS(B^n)} \frac{1}{|S_n|} \sum_{\pi \in S_n} \tr \left[ U_{A^nB^n}(\pi)\, Q_{A^nB^n}^n\, (\tau_A^{\otimes n} \otimes \sigma_{B^n}) U_{A^nB^n}(\pi)^{\dagger} \right] \\
    &= \max_{\sigma_{B^n} \in \cS(B^n)} \tr \left[ Q_{A^nB^n}^n\, (\tau_A^{\otimes n} \otimes \widetilde{\sigma}_{B^n}) \right] ,
  \end{align}
  where $\widetilde{\sigma}_{B^n} := \frac{1}{|S_n|} \sum_{\pi \in S_n} U_{B^n}(\pi) \sigma_{B^n} U_{B^n}(\pi)^{\dagger}$ is permutation invariant. Lemma~\ref{lm:universal} then yields
  \begin{align}
    \beta(Q_{A^nB^n}^n; \tau_A^{\otimes n}) \leq g_{n,d} \tr \left[ Q_{A^nB^n}^n (\tau_A^{\otimes n} \otimes \omega_{B^n}^n ) \right] ,
  \end{align}
  where we set $d = |B|$.
  Furthermore, using Audenaert \emph{et al.}'s inequality~\eqref{eq:audenaert} we find
    \begin{align}
    \beta(Q_{A^nB^n}^n; \tau_A^{\otimes n}) &\leq g_{n,d} \exp(-s \lambda_n) \tr \left[ \big( \rho_{AB}^{\otimes n} \big)^{s} \big(\tau_A^{\otimes n} \otimes \omega_{B^n}^n \big)^{1-s} \right]  \\
    &= g_{n,d} \exp\left(-s \lambda_n\right) \exp\Big( -(1-s) D_{s} \big(\rho_{AB}^{\otimes n} \,\big\| \tau_A^{\otimes n} \otimes \omega_{B^n}^n \big) \Big) \label{eq:ah-mid} \\
    &\leq g_{n,d} \exp(-s \lambda_n) \exp \Big( - (1 - s)\, I_{s}\big(\rho_{AB}^{\otimes n} \big\| \tau_A^{\otimes n}\big) \Big) 
  \end{align}
for any $\lambda_n$.
  Observing that $I_{s}\big(\rho_{AB}^{\otimes n} \big\| \tau_A^{\otimes n}\big) = n\,I_{s}(\rho_{AB} \| \tau_A)$ due to the additivity of the mutual information established in Lemma~\ref{lm:add-old}, we find that the choice
  \begin{align}
    \lambda_n = \frac{1}{s} \Big( \log g_{n,d} + n \big( R - (1-s) I_{s}(\rho_{AB} \| \tau_A) \big)\Big) \label{eq:the-lambda-n}
  \end{align}
  achieves the desired bound $\beta(Q_{A^nB^n}^n; \tau_A^{\otimes n}) \leq \exp(-nR)$.
  
  On the other hand, again using~\eqref{eq:audenaert} and Lemma~\ref{lm:add-old}, we find
  \begin{align}
    \alpha(Q_{A^nB^n}^n; \rho_{AB}^{\otimes n}) &= \tr \Big[ \big\{ \rho_{AB}^{\otimes n} < \exp(\lambda_n) \, \tau_{A}^{\otimes n} \otimes \omega_{B^n}^n \big\} \rho_{AB}^{\otimes n} \Big] \\
    &\leq \exp\big( (1-s) \lambda_n \big) \tr \Big[ \big( \rho_{AB}^{\otimes n} \big)^s \big(\tau_{A}^{\otimes n} \otimes \omega_{B^n}^n \big)^{1-s} \Big] \\
    &\leq \exp \Big( (1-s) \lambda_n - n (1-s) I_s(\rho_{AB}\|\tau_A) \Big) 
  \end{align}
  Substituting~\eqref{eq:the-lambda-n} for $\lambda_n$, we thus find
  \begin{align}
    \hat{\alpha}\Big(\!\exp(-n R); \rho_{AB}^{\otimes n} \Big\| \tau_A^{\otimes n} \!\Big) \leq \alpha(Q_{A^nB^n}^n; \rho_{AB}^{\otimes n}) \leq \exp \!\bigg( \frac{1-s}{s} \Big( \log g_{n,d} + n R -
    n I_s(\rho_{AB}\|\tau_A) \! \Big) \!\! \bigg) .
  \end{align}
  Since $\log g_{n,d} = O(\log n)$, taking the limit $n \to \infty$ yields
    \begin{align}
     \liminf_{n \to \infty} \left\{ - \frac{1}{n} \log \hat{\alpha}\Big(\exp(-n R); \rho_{AB}^{\otimes n}\Big\| \tau_A^{\otimes n}\Big) \right\} \geq \frac{1-s}{s} \left( I_{s}(\rho_{AB}, \tau_A) - R \right) .
  \end{align}
  
  Finally, since this derivation holds for all $s \in (0,1)$, we established the direct part.
\end{proof}

\subsection{Proof of Optimality}

To show optimality, we will directly employ the converse of the quantum Hoeffding bound established in~\cite{nagaoka06} together with a minimax theorem derived in Appendix~\ref{app:minimax}.

Recall that it remains to show that for any $R > 0$, we have
\begin{align}
\limsup_{n \to \infty} \left\{ - \frac{1}{n} \log \hat{\alpha}\Big(\!\exp(-n R); \rho_{AB}^{\otimes n} \Big\| \tau_A^{\otimes n}\Big) \right\} \leq \sup_{s \in (0, 1)} \left\{ \frac{1-s}{s} \big( I_{s}(\rho_{AB}\| \tau_A) - R \big) \right\}. \label{eq:hoeffding-thm3}
\end{align}

\begin{proof}[Proof of Optimality  in Theorem~\ref{th:hoeffding}]
We fix $\sigma_B \in \cS(B)$ and note that
\begin{align}
  \hat{\alpha}\Big(\exp(-n R); \rho_{AB}^{\otimes n} \Big\| \tau_A^{\otimes n}\Big) \geq \hat{\alpha}\Big(\exp(-n R); \rho_{AB}^{\otimes n} \Big\| \tau_A^{\otimes n} \otimes \sigma_B^{\otimes n} \Big)
\end{align} 
  At this point we can apply the converse of the quantum Hoeffding bound~\cite{nagaoka06} to the expression on the right-hand side, which yields
  \begin{align}
    &\limsup_{n \to \infty} \left\{ - \frac{1}{n} \log \hat{\alpha}\Big(\exp(-n R); \rho_{AB}^{\otimes n}\Big\| \tau_A^{\otimes n}\Big) \right\} \label{eq:lhs1} \\
    &\qquad \qquad \leq \limsup_{n \to \infty} \left\{ - \frac{1}{n} \log \hat{\alpha}\Big(\exp(-n R); \rho_{AB}^{\otimes n} \Big\| \big(\tau_A \otimes \sigma_B)^{\otimes n}\Big) \right\}\\
    &\qquad \qquad \leq \sup_{s \in (0,1)} \left\{ \frac{1-s}{s} \left( D_{s}(\rho_{AB} \| \tau_A \otimes \sigma_B) - R \right) \right\} .
  \end{align}
  Since this holds for all $\sigma_B \in \cS(B)$, the limit in~\eqref{eq:lhs1} is in fact upper bounded by
  \begin{align}
    \inf_{\sigma_B \in \cS(B)} \sup_{s \in (0,1)} \left\{ \frac{1-s}{s} \left( D_{s}(\rho_{AB} \| \tau_A \otimes \sigma_B) - R \right) \right\} \label{eq:minimax-missing}
  \end{align}
  The infimum can be restricted to the convex subset $\cS_{\rho}(B)$ of $\cS(B)$ of operators in the support of $\rho_B$. 
  It remains to observe that the functional
  \begin{align}
    f: (s, \sigma_B) \mapsto (1-s) D_s(\rho_{AB}\|\tau_A \otimes \sigma_B) = -\log \tr \big[\rho_{AB}^s (\tau_A \otimes \sigma_B)^{1-s}\big]
  \end{align}
  is convex in $\sigma_B$ for $s \in (0, 1)$ due to the operator concavity of $t \mapsto t^{1-s}$
  and concave in $s$ as was shown in~\cite[Lem.~2.1]{audenaert12}. 
Since $\cS_{\rho}(B)$ is convex, $f$ is also $\frac{1}{2}$-convexlike. Finally, note that $f(s, \sigma_B)$ is finite for $s \in[0,1]$.
   Hence, the minimax theorem (Proposition~\ref{pr:minimax}) in Appendix~\ref{app:minimax} applies to~\eqref{eq:minimax-missing}. This, together with the definition of $I_{\alpha}$ in~\eqref{eq:ialpha}, concludes the proof.
\end{proof}


\section{Strong Converse Exponent}
\label{sc:sc}


Our second result considers the case where the error of the second kind goes to zero exponentially with a rate exceeding the mutual information $I(\rho_{AB}\|\tau_A)$. In this case, we find that the error of the first kind converges to $1$ exponentially fast, and the exponent is determined by the sandwiched R\'enyi mutual information, $\widetilde{I}_{s}(\rho_{AB}\|\tau_A)$, with $s > 1$.

\begin{theorem}\label{th:han}
  Let $\rho_{AB} \in \cS(AB)$ and $\tau_{A} \in \cS(A)$ with $\tau_A \gg \rho_A$. Then, for any $0 < R < \widetilde{I}_{\infty}(\rho_{AB}\|\tau_A)$, we have
  \begin{equation}
    \lim_{n \to \infty} \left\{ - \frac{1}{n} \log \bigg( 1 - \hat{\alpha}\Big(\!\exp(-n R); \rho_{AB}^{\otimes n} \Big\| \tau_A^{\otimes n}\Big) \bigg) \right\} = \sup_{s > 1} \left\{ \frac{s-1}{s} \left( R - \widetilde{I}_{s}(\rho_{AB}\| \tau_A) \right) \right\}. \label{eq:scexp}
  \end{equation}
\end{theorem}

Note that the range of $R$ for which this result is valid can be extended, but here we restrict our attention to the range where $R$ is sufficiently close to $I(\rho_{AB}\|\tau_A)$. For general $R > 0$ we refer the reader to a recent analysis of the strong converse exponent by Mosonyi and Ogawa~\cite{mosonyi14} that can be adapted to cover the situation at hand here. (See also~\cite{nakagawa93} for an earlier discussion of this issue in classical hypothesis testing.)


Again, we are interested in the following two special cases:
\begin{corollary}
    Let $\rho_{AB} \in \cS(AB)$. Then, for suitable $R > 0$ (see Theorem~\ref{th:han}), we have
  \begin{align}
    \lim_{n \to \infty} \left\{ - \frac{1}{n} \log \bigg( 1 - \hat{\alpha}\Big(\!\exp(-n R); \rho_{AB}^{\otimes n} \Big\| \rho_A^{\otimes n}\Big) \bigg) \right\} &= \sup_{s > 1} \left\{ \frac{s-1}{s} \big( R - \widetilde{I}_{s}(A\!:\!B^{\downarrow})_{\rho}  \big) \right\}, \\
        \lim_{n \to \infty} \left\{ - \frac{1}{n} \log \bigg( 1 - \hat{\alpha}\Big(\!\exp(-n R); \rho_{AB}^{\otimes n} \Big\| \pi_A^{\otimes n}\Big) \bigg) \right\} &= \sup_{s > 1} \left\{ \frac{s-1}{s} \big( R - \log |A| + \widetilde{H}_{s}^{\uparrow}(A|B)_{\rho} \big) \right\} .
  \end{align}
\end{corollary}
This corollary establishes an operational interpretation of the sandwiched R\'enyi mutual information, $\widetilde{I}_{\alpha}^{\downarrow}(A\!:\!B)_{\rho}$, as well as the sandwiched R\'enyi conditional entropy, $\widetilde{H}_{\alpha}^{\uparrow}(A|B)_{\rho}$, for $\alpha > 1$.

Before we commence with the proof, we will discuss the classical Neyman-Pearson test we use and some results from classical large deviation theory. Following this, we treat the proof of the achievability and optimality in Theorem~\ref{th:han} separately. 

\subsection{Classical Neyman-Pearson Test}

To show the direct part we employ a classical Neyman-Pearson test for the pinched state $\cP_{\tau_A^{\otimes n} \otimes\, \omega_{B^n}^n} ( \rho_{AB}^{\otimes n} )$ and the state $\tau_{A}^{\otimes n} \otimes \omega_{B^n}^n$. The idea to use a classical Neyman-Pearson test on the pinched state goes back to~\cite{hayashi02b}.
%
We start by discussing some properties of this test.

\begin{lemma}
  \label{lm:test}
  Let $\rho_{AB} \in \cS(AB)$, $\tau_{A} \in \cS(A)$, $n \in \mathbb{N}$ and $\mu_{n} \in \mathbb{R}$. Consider the test
\begin{align}
     Q_{A^nB^n}^n := \left\{ \cP_{\tau_A^{\otimes n} \otimes\, \omega_{B^n}^n} \big( \rho_{AB}^{\otimes n} \big) \geq \exp(\mu_{n}) \, \tau_{A}^{\otimes n} \otimes \omega_{B^n}^n \right\} . \label{eq:test-tilde}
\end{align}
Let $\{ |\phi_{x_n}\rangle \}_{x_n}$ be a common orthonormal eigenbasis of $ \cP_{\tau_A^{\otimes n} \otimes\, \omega_{B^n}^n} ( \rho_{AB}^{\otimes n} )$ and $\tau_A^{\otimes n} \otimes \omega_{B^n}^n$ and define the probability distributions 
 \begin{align}
   P_n(x_n) = \big\langle \phi_{x_n} \big| \cP_{\tau_A^{\otimes n} \otimes\, \omega_{B^n}^n} \big( \rho_{AB}^{\otimes n} \big)
   \big| \phi_{x_n} \big\rangle 
   , \quad \textrm{and} \quad
   Q_n(x_n) = \big\langle \phi_{x_n} \big|\tau_A^{\otimes n} \otimes \omega_{B^n}^n \big| \phi_{x_n} \big\rangle \,. \label{eq:probs2}
 \end{align}
  Then, with $X_n$ distributed according to the law $P_n$ and $X_n'$ according to the law $Q_n$, we have
  \begin{align}
    \alpha(Q_{A^nB^n}^n; \rho_{AB}^{\otimes n}) &= \Pr \left[ P_n(X_n) < \exp(\mu_n) Q_n(X_n) \right] \,, \quad \qquad \textrm{and} \label{eq:alphabound} \\
    \beta(Q_{A^nB^n}^n; \tau_A^{\otimes n}) &\leq g_{n,d} \Pr \left[ P_n(X_n') \geq \exp(\mu_n) Q_n(X_n') \right] \label{eq:betabound} ,
  \end{align}
  where $d = |B|$.
\end{lemma}

\begin{proof}
  It is easy to verify that the pinched quantity is permutation invariant, and thus $Q_{A^nB^n}^n$ is permutation invariant as well.
Let us evaluate
  \begin{align}
    \beta(Q_{A^nB^n}^n; \tau_A^{\otimes n}) &= \max_{\sigma_{B^n} \in \cS(B^n)} \tr \left[ Q_{A^nB^n}^n\, (\tau_A^{\otimes n} \otimes \sigma_{B^n}) \right] \\
    &= \max_{\sigma_{B^n} \in \cS(B^n)} \frac{1}{|S_n|} \sum_{\pi \in S_n} \tr \left[ U_{A^nB^n}(\pi)\, Q_{A^nB^n}^n\, (\tau_A^{\otimes n} \otimes \sigma_{B^n}) U_{A^nB^n}(\pi)^{\dagger} \right] \\
    &= \max_{\sigma_{B^n} \in \cS(B^n)} \tr \left[ Q_{A^nB^n}^n\, (\tau_A^{\otimes n} \otimes \widetilde{\sigma}_{B^n}) \right] ,
  \end{align}
  where $\widetilde{\sigma}_{B^n} := \frac{1}{|S_n|} \sum_{\pi \in S_n} U_{B^n}(\pi) \sigma_{B^n} U_{B^n}(\pi)^{\dagger}$ is permutation invariant. Lemma~\ref{lm:universal} then yields
  \begin{align}
    \beta(Q_{A^nB^n}^n; \tau_A^{\otimes n}) &\leq g_{n,d} \tr \left[ Q_{A^nB^n}^n (\tau_A^{\otimes n} \otimes \omega_{B^n}^n ) \right] \\
    &= g_{n,d} \tr \left[ \left\{ \cP_{\tau_A^{\otimes n} \otimes\, \omega_{B^n}^n} \big( \rho_{AB}^{\otimes n} \big) \geq \exp(\mu_{n}) \, \tau_{A}^{\otimes n} \otimes \omega_{B^n}^n \right\} (\tau_A^{\otimes n} \otimes \omega_{B^n}^n ) \right] .
    \label{eq:commute}
  \end{align}
  
 The two operators $ \cP_{\tau_A^{\otimes n} \otimes\, \omega_{B^n}^n} ( \rho_{AB}^{\otimes n} )$ and $\tau_A^{\otimes n} \otimes \omega_{B^n}^n$ in~\eqref{eq:commute} commute. Let $\{ |\phi_{x_n}\rangle \}_{x_n}$ be a common orthonormal eigenbasis for these operators and define the probability distributions in~\eqref{eq:probs2} as well as the corresponding random variables.
 Then we can simplify~\eqref{eq:commute} by noting that
 \begin{align}
   \tr \left[ \left\{ \cP_{\tau_A^{\otimes n} \otimes\, \omega_{B^n}^n} \big( \rho_{AB}^{\otimes n} \big) \geq \exp(\mu_{n}) \, \tau_{A}^{\otimes n} \otimes \omega_{B^n}^n \right\} (\tau_A^{\otimes n} \otimes \omega_{B^n}^n ) \right] = \Pr \left[ P_n(X_n') \geq \exp(\mu_n) Q_n(X_n') \right] ,
 \end{align}
  which yields~\eqref{eq:betabound}.
  Finally, it is easy to verify that 
  \begin{align}
    \alpha(Q_{A^nB^n}^n; \rho_{AB}^{\otimes n}) &= \tr \Big[ \left\{ \cP_{\tau_A^{\otimes n} \otimes\, \omega_{B^n}^n} \big( \rho_{AB}^{\otimes n} \big) < \exp(\mu_n) \, \tau_{A}^{\otimes n} \otimes \omega_{B^n}^n \right\} \cP_{\tau_A^{\otimes n} \otimes\, \omega_{B^n}^n} \big( \rho_{AB}^{\otimes n} \big) \Big] \\
   &= \Pr \left[ P_n(X_n) < \exp(\mu_n) Q_n(X_n) \right] .
   \end{align}
\end{proof}

\subsection{Classical Large Deviation Theory}

Our proof will rely on a variant of the G\"artner-Ellis theorem of large deviation theory (see, e.g.,~\cite[Sec.~2 and Sec.~3.4]{dembo98} for an overview), which we recall here. Given a sequence of random variables $\{ Z_n \}_{n \in \mathbb{N}}$ we
introduce its asymptotic \emph{cumulant generating function} as
\begin{align}
   \Lambda_Z(t) &:= \lim_{n \to \infty} \left\{ \frac{1}{n} \log \big( \Exp \left[ \exp ( n t Z_n ) \right] \big) \right\} ,
   \label{eq:cum}
\end{align}
if it exists.
For our purposes it is sufficient to use the following variant of the G\"artner-Ellis theorem due to Chen~\cite[Thm.~3.6]{chen00} (see also~\cite[Lem.~B.2]{mosonyi14}).
\begin{proposition}
  \label{pr:mosonyi}
  Let us assume that $t \mapsto \Lambda_Z(t)$ as defined in~\eqref{eq:cum} exists and is differentiable in some interval $(a, b)$. Then, for any $z \in \big( \lim_{t \searrow a} \Lambda_Z'(t), \lim_{t \nearrow b} \Lambda_Z'(t) \big)$, we have
  \begin{align}
    \limsup_{n \to \infty} \left\{ - \frac{1}{n} \log \Pr [ Z_n \geq z] \right\} \leq \sup_{t \in (a, b)} \left\{ z t - \Lambda_Z(t) \right\} \,.
  \end{align}
\end{proposition}

Finally, in order to evaluate the asymptotic cumulant generating function, we will employ the asymptotic achievability in Proposition~\ref{pr:universal}, namely the fact that
\begin{align}
  \lim_{n\to\infty} \frac{1}{n} D_{\alpha}(P_n\|Q_n) = \widetilde{I}_{\alpha}(\rho_{AB}\|\tau_A) \label{eq:universal-pnqn} .
\end{align}

\subsection{Proof of Achievability}

We are now ready to present the proof of achievability, namely we show that
  \begin{equation}
    \limsup_{n \to \infty} \left\{ - \frac{1}{n} \log \bigg( 1 - \hat{\alpha}\Big(\!\exp(-n R); \rho_{AB}^{\otimes n} \Big\| \tau_A^{\otimes n}\Big) \bigg) \right\} \leq \sup_{s > 1} \left\{ \frac{s-1}{s} \left( R - \widetilde{I}_{s}(\rho_{AB}\| \tau_A) \right) \right\}. \label{eq:scproofthis}
  \end{equation}
We restrict our attention to the case where $I(\rho_{AB}\|\tau_A) < R < \widetilde{I}_{\infty}(\rho_{AB}\|\tau_A)$, for which we provide a novel proof. 

\begin{proof}[Proof of Achievability in Theorem~\ref{th:han}]
Given an arbitrary fixed $s \in (1, \infty)$,
we choose a sequence $\{\mu_{n} \}_{n \in \mathbb{N}}$ of real numbers 
as
\begin{align}
    \mu_n = \frac{1}{s} \Big( \log g_{n,d} + n R + (s - 1) D_{s} (P_n \| Q_n) \Big) . \label{eq:sc-mu22}
 \end{align}
Consider the sequence of tests given by Lemma~\ref{lm:test}. Then, due to~\eqref{eq:betabound}, we have
\begin{align}
    \beta(Q_{A^nB^n}^n; \tau_A^{\otimes n}) &\leq g_{n,d} \Pr \left[ P_n(X_n') \geq \exp(\mu_n) Q_n(X_n') \right] 
    \\
    &\leq g_{n,d} \exp(-s \mu_n) \sum P_n(x)^{s} Q_n(x)^{1-s} \\
    &= g_{n,d} \exp\left(-s \mu_n\right) \exp\left( (s-1) D_{s}( P_n \| Q_n ) \right) . 
\end{align}
Hence, the requirement that $\beta(Q_{A^nB^n}^n; \tau_A^{\otimes n}) \leq \exp(-n R)$ can then be satisfied by the choice \eqref{eq:sc-mu22}.
Let us now take a closer look at the error of the first kind in~\eqref{eq:alphabound}. We find
 \begin{align}
   1 - \alpha(Q_{A^nB^n}^n; \rho_{AB}^{\otimes n}) &= \Pr \left[ P_n(X_n) \geq \exp(\mu_n) Q_n(X_n) \right] = \Pr \left[ Z_n \geq 0 \right] , 
\label{H2}
 \end{align}
 where we introduced the sequence of random variables $\{Z_n \}_{n \in \mathbb{N}}$ with 
 \begin{align}
   Z_n(X_n) :&= \frac{1}{n} \Big( \log P_n(X_n) - \log Q_n(X_n) - \mu_n \Big)  \\
   &= \frac{1}{n} \Big( \log P_n(X_n) - \log Q_n(X_n) - \frac{\log g_{n,d}}{s} - \frac{s-1}{s} D_s(P_n\|Q_n) \Big) - \frac{R}{s} \,.
   \end{align}
Since $\beta(Q_{A^nB^n}^n; \tau_A^{\otimes n}) \leq \exp(-n R)$ holds for our test,
\eqref{H2} yields 
  \begin{align}
 1- \hat{\alpha}\Big(\exp(-n R); \rho_{AB}^{\otimes n}\Big\| \tau_A^{\otimes n}\Big) 
&\geq  \Pr \left[ Z_n \geq 0\right] ,
\label{H6}
   \end{align}
which implies that
 \begin{align}
    \limsup_{n \to \infty} \left\{ - \frac{1}{n} \log \bigg( 1- \hat{\alpha}\Big(\exp(-n R); \rho_{AB}^{\otimes n}\Big\| \tau_A^{\otimes n}\Big) \bigg) \right\}  &\leq \limsup_{n\to\infty} \left\{ - \frac{1}{n} \log \big( \Pr \left[ Z_n \geq 0 \right] \big) \right\} .
\label{H4}
\end{align}

Next, let us introduce the function
 \begin{align}
   f: (s, t) \mapsto t \big( R - s \widetilde{I}_{1+t}(\rho_{AB}\|\tau_A) + (s-1) \widetilde{I}_s(\rho_{AB}\|\tau_A) \big) .
  \end{align}
We want to tackle the asymptotics of~\eqref{H4} using the G\"artner-Ellis theorem. We therefore calculate the asymptotic cumulant generating function, as in~\eqref{eq:cum}, for $t \geq -\frac12$ as follows:
  \begin{align}
   \Lambda_Z(t) &= \lim_{n \to \infty} \left\{ \frac{1}{n} \log \big( \Exp \left[ \exp ( n t Z_n ) \right] \big) \right\} \\
   &=
     \lim_{n \to \infty} \left\{ \frac{1}{n} \log \Exp \left[ \frac{P_n(X_n)^t}{Q_n(X_n)^t} \right]  - \frac{t \log g_{n,d}}{sn} - \frac{t(s - 1)}{sn} D_s(P_n\|Q_n) \right\} - t \frac{R}{s} \\
   &= t \lim_{n \to \infty} \left\{ \frac{1}{n} D_{1+t}(P_n \| Q_n) - \frac{s - 1}{ns} D_s(P_n\|Q_n) \right\} - t \frac{R}{s} \\
   &= t \left( \widetilde{I}_{1+t}(\rho_{AB}\|\tau_A) - \frac{s-1}{s} \widetilde{I}_s(\rho_{AB}\|\tau_A) - \frac{R}{s} \right) \\
   &= - \frac{f(s,t)}{s} \label{eq:use-universal-here} \,.
 \end{align}
 We used Proposition~\ref{pr:universal} in the form of~\eqref{eq:universal-pnqn} twice to establish~\eqref{eq:use-universal-here}. Now, Proposition~\ref{pr:diff} ensures that $\Lambda_Z(t)$ is continuously differentiable for $t \geq -\frac12$. Furthermore, we have
 \begin{align}
   \lim_{t \to 0} \Lambda_Z'(t)  &= I(\rho_{AB}\|\tau_A) - \frac{s-1}{s} \widetilde{I}_{s}(\rho_{AB}\|\tau_A)  - \frac{R}{s} \\
   &\leq I(\rho_{AB}\|\tau_A) - \frac{s-1}{s} I(\rho_{AB}\|\tau_A)  - \frac{R}{s} = \frac{1}{s} \left( I(\rho_{AB}\|\tau_A) - R \right) < 0 \,,
 \end{align}
 where we used that $R > I(\rho_{AB}\|\tau_A)$ in the last step. 

 On the other hand, using the convexity of $t \mapsto \phi(t) := t \widetilde{I}_{1+t}(\rho_{AB}\|\tau_A)$ (cf.~Corollary~\ref{corr:point-wise}) and $\phi(0) = 0$, we find $\phi(\lambda t) \leq \lambda \phi(t)$ for all $\lambda \in (0, 1)$. And, thus, $\phi'(t) = \lim_{\lambda \to 1} \frac{\phi(t) - \phi(\lambda t)}{t (1-\lambda)} \geq \frac{\phi(t)}{t}$. Let now $t_0$ be such that $R < \widetilde{I}_{t_0+1}(\rho_{AB}\|\tau_A)$. Then, for any $s \leq t_0 + 1$, we find
 \begin{align}
   \lim_{t \to t_0} \Lambda_Z'(t) &\geq \widetilde{I}_{t_0+1}(\rho_{AB}\|\tau_A) - \frac{s-1}{s} \widetilde{I}_{s}(\rho_{AB}\|\tau_A) - \frac{R}{s} \\
    &\geq \widetilde{I}_{t_0+1}(\rho_{AB}\|\tau_A) - \frac{s-1}{s} \widetilde{I}_{t_0+1}(\rho_{AB}\|\tau_A) - \frac{R}{s} = \frac{1}{s}\big( \widetilde{I}_{t_0+1}(\rho_{AB}\|\tau_A) - R ) > 0 \,, 
\end{align}
where we used that $R < \widetilde{I}_{t_0+1}(\rho_{AB}\|\tau_A)$.
Hence, we may apply Proposition~\ref{pr:mosonyi}, which yields
 \begin{align}
   \limsup_{n\to\infty} \left\{ - \frac{1}{n} \log \big( \Pr \left[ Z_n \geq 0 \right] \big) \right\} &\leq \sup_{0 < t < t_0} \left\{ - \Lambda_Z(t) \right\} 
\leq \sup_{0 \leq t \leq t_0} \, \frac{f(s, t)}{s}\,.
\label{H3}
 \end{align}  
Since the above holds for all $s \in (1, t_0+1]$, we indeed find
\begin{align}
&\limsup_{n \to \infty} \left\{ - \frac{1}{n} \log \bigg( 1- \hat{\alpha}\Big(\exp(-n R); \rho_{AB}^{\otimes n}\Big\| \tau_A^{\otimes n}\Big) \bigg) \right\} \leq \inf_{1 < s \leq t_0+1} \sup_{0 \leq t \leq t_0} \, \frac{f(s, t)}{s} ,
 \end{align}
 It is straightforward to verify that $f(s, t)$ is concave in $t$ and convex in $s$ since $t \mapsto t \widetilde{I}_{1+t}(\rho_{AB}\|\tau_A)$ is convex (cf.~Corollary~\ref{corr:point-wise}). Moreover, we optimize $t$ over a compact convex set and $s$ over a convex set. Thus, by Proposition~\ref{pr:minimax} in Appendix~\ref{app:minimax} applied to $-f$, we have
 \begin{align}
    \inf_{1 < s \leq t_0+1} \sup_{0 \leq t \leq t_0} \frac{f(s, t)}{s} &=  \sup_{0 \leq t \leq t_0} \inf_{1 < s \leq t_0+1} \frac{f(s, t)}{s} \\
    &= \sup_{0 < t \leq t_0} \inf_{1 < s \leq t_0+1} \frac{f(s, t)}{s} \\
    &\leq \sup_{0 < t \leq t_0} \frac{f(t+1,t)}{t+1} \leq \sup_{0 < t} \frac{f(t+1,t)}{t+1} \,,
 \end{align}
 where we restricted the optimization to strictly positive $t$ since the inner term vanishes for $t = 0$ and we simply chose $s = t+1$ in the penultimate step. 
 The resulting term corresponds to the right hand side of~\eqref{eq:scproofthis} by a suitable change of variable,
which concludes the proof.
\end{proof}

\subsection{Proof of Optimality}

It remains to show that, for all $R > 0$,
  \begin{equation}
    \liminf_{n \to \infty} \left\{ - \frac{1}{n} \log \bigg( 1 - \hat{\alpha}\Big(\!\exp(-n R); \rho_{AB}^{\otimes n} \Big\| \tau_A^{\otimes n}\Big) \bigg) \right\} \geq \sup_{s > 1} \left\{ \frac{s-1}{s} \left( R - \widetilde{I}_{s}(\rho_{AB}\| \tau_A) \right) \right\}. \label{eq:scproofthis2}
  \end{equation}

\begin{proof}[Proof of Optimality in Theorem~\ref{th:han}]
  Analogous to the optimality proof for Theorem~\ref{th:hoeffding}, we first fix $\sigma_B \in \cS(B)$ and this time apply the converse bound in~\cite[Thm.~IV.9]{mosonyiogawa13}. This yields
  \begin{align}
    &\liminf_{n \to \infty} \left\{ - \frac{1}{n} \log \bigg( 1- \hat{\alpha}\Big(\exp(-n R); \rho_{AB}^{\otimes n}\Big\| \tau_A^{\otimes n}\Big) \bigg) \right\}  \\
    &\qquad \qquad \geq \liminf_{n \to \infty} \left\{ - \frac{1}{n} \log \bigg( 1- \hat{\alpha}\Big(\exp(-n R); \rho_{AB}^{\otimes n} \Big\| \tau_A^{\otimes n} \otimes \sigma_B^{\otimes n} \Big) \bigg) \right\}\\
    &\qquad \qquad \geq \sup_{s > 1} \left\{ \frac{s-1}{s} \big( R - \widetilde{D}_{s}(\rho_{AB} \| \tau_A \otimes \sigma_B)  \big) \right\} \label{eq:rhs2}.
  \end{align}
  Since this holds for all $\sigma_B \in \cS(B)$, we may maximize the expression in~\eqref{eq:rhs2} with regards to $\sigma_B$, yielding the desired result.
\end{proof}


\section{Stein's Lemma and Second Order}
\label{sc:second}

As a direct consequence of our results on the error exponents (Hoeffding bound) and strong converse exponents, we find the following variant of Stein's lemma and its strong converse.

\begin{corollary}
Let $\rho_{AB} \in \cS(AB)$ and $\tau_{A} \in \cS(A)$ with $\tau_A \gg \rho_A$. Then,
\begin{align}
 \sup \bigg\{ R \in \mathbb{R} : \lim_{n \to \infty} \hat{\alpha}\Big( \exp(-n R); \rho_{AB}^{\otimes n} \Big\| \tau_A^{\otimes n} \Big) = 0 \bigg\} &= \\
 \inf \bigg\{ R \in \mathbb{R} : \lim_{n \to \infty} \hat{\alpha}\Big( \exp(-n R); \rho_{AB}^{\otimes n} \Big\| \tau_A^{\otimes n} \Big) = 1 \bigg\} &= I(\rho_{AB}\|\tau_A) \,.
\end{align}
\end{corollary}

For completeness, we also investigate the second order behavior, namely we investigate the error of the first kind when the error of the second kind vanishes as $\exp\big(- n I(\rho_{AB}\|\tau_A) - \sqrt{n}\,r\big)$.
This analysis takes us a step beyond quantum Stein's lemma~\cite{hiai91,ogawa00}.
Paralleling the results in~\cite{tomamichel12,li12} for simple hypothesis tests, we find that the error of the first kind converges to a constant.

\begin{theorem}
  \label{th:second}
Let $\rho_{AB} \in \cS(AB)$ and $\tau_{A} \in \cS(A)$ with $\tau_A \gg \rho_A$. Then, for any $r \in \mathbb{R}$, we have
  \begin{align}
    \lim_{n \to \infty} \left\{ \hat{\alpha}\Big( \exp \big( -n I(\rho_{AB}\|\tau_A) - \sqrt{n} \,r\log e\big); \rho_{AB}^{\otimes n} \Big\| \tau_A^{\otimes n} \Big) \right\} = \Phi \left( \frac{r\log e}{\sqrt{V(\rho_{AB}\|\tau_A)}} \right) ,
  \end{align}
  where $\Phi$ is the cumulative standard normal (Gaussian) distribution.
\end{theorem}
Clearly one can simplify the expression by substituting $r' = r \log e$, but the advantage of the above representation is that we see that both sides of the equality are independent of the choice of base of the logarithm.

\subsection{Proof of Achievability}

Let $M_{X}(t) := \mathbb{E}\big[ e^{t X} \big]$ denote the moment generating function of a real random variable $X$. We need the following property of moment generating functions, a variant of L\'evi's continuity theorem~\cite[Thm.~2]{mukherjea06}.
\begin{lemma}\label{lm:moments}
   Let $0 < a < b$. If a sequence of random variables $\{ X_n \}_{n \in \mathbb{N}}$ satisfies $\lim_{n \to \infty} M_{X_n}(t) = M_X(t)$ for some random variable $X$ and all $t \in (a, b)$, then $\lim_{n \to \infty} \Pr[X_n \leq k] = \Pr[X \leq k]$ for all $k \in \mathbb{R}$.
\end{lemma}

Second order achievability is now proven using the hypothesis test of Section~\ref{sc:sc} together with Corollary~\ref{corr:2nd-order} in the following.

\begin{proof}[Proof of Achievability in Theorem~\ref{th:second}]
We again use the test in Lemma~\ref{lm:test} and set $\mu_n = n I(\rho_{AB}\|\tau_A) + \sqrt{n}\,r \log e + \log g_{n,d}$.
Then, Eq.~\eqref{eq:betabound} yields
\begin{align}
  \beta(Q_{A^nB^n}^n; \tau_A^{\otimes n}) &\leq g_{n, d_B} \Pr \left[ P_n(X_n') \geq \exp(\mu_n) Q_n(X_n') \right] \\
  &\leq g_{n, d_B} \exp(-\mu_n) \Pr \left[ P_n(X_n) \geq \exp(\mu_n) Q_n(X_n) \right] \\
  &\leq 
  \exp \big( - n I(\rho_{AB}\|\tau_A) - \sqrt{n}\,r \log e \big) \,. \label{eq:second-kind-ok}
\end{align}
Moreover, using~\eqref{eq:alphabound}, we find
\begin{align}
   \alpha(Q_{A^nB^n}^n; \rho_{AB}^{\otimes n}) &= \Pr \left[ \log P_n(X_n)  - \log Q_n(X_n) < n I(\rho_{AB}\|\tau_A) + \sqrt{n}\,r\log e + \log g_{n,d} \right] \\
   &= \Pr \left[ Y_n(X_n) < r  \right] . \label{eq:sec1}
\end{align}
where we defined the following sequence of random variables:
\begin{align}
  Y_n :=
   \frac{1}{\sqrt{n}} \left( \ln P_n(X_n) - \ln Q_n(X_n) - n \frac{I(\rho_{AB}\|\tau_A)}{\log e} - \ln g_{n, d_B}\right) \,.
\end{align}
with $X_n$ distributed according to the law $P_n$ as usual.

Now, note that the moment generating function of the sequence $\{ Y_n \}_n$, denoted $M_Y(t)$, can be evaluated using the fact that
\begin{align}
  \ln M_Y(t) =& \lim_{n \to \infty} \big\{  \ln \Exp \left[ e^{ t Y_n } \right] \big\} \\
=& \lim_{n \to \infty} \left\{   \frac{t}{\sqrt{n} \log e} \big(  D_ {1 + \frac{t}{\sqrt{n}}}(P_n \| Q_n) - n I(\rho_{AB}\|\tau_A) \big)\right\} - \lim_{n \to \infty} \left\{ \frac{t}{\sqrt{n}} \ln g_{n,d} \right\} \\
=& \frac{t^2}{2 (\log e)^2} V(\rho_{AB}\|\tau_A) \,.
\end{align}
In the last step we used Corollary~\ref{corr:2nd-order} to evaluate the first term and the fact that $\log g_{n,d} = O(\log n)$ to evaluate the second term. Hence, by Lemma~\ref{lm:moments}, the sequence of random variable $\{ Y_n \}_n$ converges in distribution to a random variable $Y$ with cumulant generating function $\ln M_Y(t)$, i.e.,
a Gaussian random variable with zero mean and variance~$\frac{1}{(\log e)^2} V(\rho_{AB}\|\tau_A)$. In particular, this yields
\begin{align}
   \lim_{n \to \infty} \Pr \left[ Y_n < r \right] = \Pr \left[ Y < r \right] = \Phi \left( \frac{r \log e}{\sqrt{V(\rho_{AB}\|\tau_A)}} \right) . \label{eq:sec2}
\end{align}

Finally, due to~\eqref{eq:second-kind-ok} we have
\begin{align}
  \limsup_{n \to \infty} \left\{ \hat{\alpha}\Big( \exp \big( -n I(\rho_{AB}\|\tau_A) - \sqrt{n}\,r \big); \rho_{AB}^{\otimes n} \Big\| \tau_A^{\otimes n} \Big) \right\} &\leq \lim_{n \to \infty} \alpha(Q_{A^nB^n}^n; \rho_{AB}^{\otimes n}) \,. 
\end{align}
Combining this with~\eqref{eq:sec1} and~\eqref{eq:sec2} concludes the proof.
\end{proof}


\subsection{Proof of Optimality}

The proof of the optimality follows directly from the bound in Eq.~\eqref{eq:hatbound2} and the second order expansion for binary quantum hypothesis testing independently established in~\cite{li12} and~\cite{tomamichel12} as follows.

\begin{proof}[Proof of Optimality in Theorem~\ref{th:second}]
The papers \cite{li12} and~\cite{tomamichel12} showed that
  \begin{align}
    \lim_{n \to \infty} \left\{ \hat{\alpha}\Big( \exp \big( -n I(\rho_{AB}\|\tau_A) - \sqrt{n}\,r \big); \rho_{AB}^{\otimes n} \Big\| \tau_A^{\otimes n}
\otimes \rho_B^{\otimes n} \Big) 
\right\} = \Phi \left( \frac{r}{\sqrt{V(\rho_{AB}\|\tau_A)}} \right) .
  \end{align}
Since 
$\hat{\alpha}\big( \exp ( -n I(\rho_{AB}\|\tau_A) - \sqrt{n}\,r ); \rho_{AB}^{\otimes n} \big\| \tau_A^{\otimes n} \big) 
\ge
\hat{\alpha}\big( \exp ( -n I(\rho_{AB}\|\tau_A) - \sqrt{n}\,r ); \rho_{AB}^{\otimes n} \big\| \tau_A^{\otimes n}
\otimes \rho_B^{\otimes n} \big)$,
we obtain the optimality.
\end{proof}





\paragraph*{Acknowledgements:}
MT thanks Mil\'an Mosonyi for enlightening discussions throughout this project, for many comments that helped improve the presentation, and for sharing his notes on the G\"artner-Ellis theorem. We also thank Mark M.~Wilde for comments. The current version of the manuscript is much improved following comments by an anonymous referee.
MH is partially supported by a MEXT Grant-in-Aid for Scientific Research (A) No. 23246071, and by the National Institute of Information and Communication Technology (NICT), Japan.
MT acknowledges support from the MOE Tier 3 Grant ``Random numbers from quantum processes'' (MOE2012-T3-1-009).
The Centre for Quantum Technologies is funded by the Singapore Ministry of Education and the National Research Foundation as part of the Research Centers of Excellence program.

\appendix

\section{A Minimax Theorem}
\label{app:minimax}

Here we show a useful minimax theorem, that is essentially a corollary of K\"onig's minimax theorem~\cite{koenig68}. First, we need to introduce a weaker notion of concavity and convexity. A function $f: \mathcal{X} \times \mathcal{Y} \to \mathbb{R}$ is \emph{$\frac{1}{2}$-concavelike} on $\mathcal{X}$ if, for every $x_1, x_2 \in \mathcal{X}$, there exists $x_3 \in \mathcal{X}$ such that
\begin{align}
  f(x_3, y) \geq \frac12 \big( f(x_1,y) + f(x_2, y) \big) \qquad \textrm{for every} \quad y \in \mathcal{Y} . \label{eq:concavelike}
\end{align}
Analogously, $f$ is \emph{$\frac{1}{2}$-convexlike} on $\mathcal{Y}$ if, for every $y_1, y_2 \in \mathcal{Y}$, there exists $y_3 \in \mathcal{Y}$ such that
\begin{align}
  f(x, y_3) \leq \frac12 \big( f(x,y_1) + f(x, y_2) \big) \qquad \textrm{for every} \quad x \in \mathcal{X} .
\end{align} 
K\"onig's minimax theorem now reads as follows~\cite{koenig68} (see also~\cite{kassay94}). Let $\mathcal{Y}$ be a compact Hausdorff space and let $f(x, \cdot)$ be lower-semicontinuous for every $x \in \mathcal{X}$. Moreover, let $f$ be $\frac{1}{2}$-concavelike on $\mathcal{X}$ and $\frac{1}{2}$-convexlike on $\mathcal{Y}$. Then, we have
\begin{align}
  \sup_{x \in \mathcal{X}} \inf_{y \in \mathcal{Y}} f(x,y) = \inf_{y \in \mathcal{Y}} \sup_{x \in \mathcal{X}} f(x,y)
  \label{eq:koenig-minimax} .
\end{align}
From this we derive the following result. 

\begin{proposition}
  \label{pr:minimax}
  Let $\mathcal{X} \subset \mathbb{R}^+$ be convex and let $\mathcal{Y}$ be a compact Hausdorff space. Further, let $f: \mathcal{X} \times \mathcal{Y} \to \mathbb{R}$ be concave on $\mathcal{X}$ as well as $\frac{1}{2}$-convexlike and lower semi-continuous on $\mathcal{Y}$. Then,
  \begin{align}
    \sup_{x \in \mathcal{X}} \inf_{y \in \mathcal{Y}} \frac{f(x,y)}{x} = \inf_{y \in \mathcal{Y}} \sup_{x \in \mathcal{X}} \frac{f(x,y)}{x}  \,.
  \end{align}
\end{proposition}

\begin{proof}
   We just need to show that $g(x,y) = \frac{f(x,y)}{x}$ satisfies the conditions required for~\eqref{eq:koenig-minimax} to hold. First, $g(x,y)$ is $\frac{1}{2}$-convexlike and lower-semicontinuous in $y$ by assumption. It is also $\frac{1}{2}$-concavelike in $x$ due to the following argument. Let $x_1, x_2 \in \mathcal{X}$ with $x_1 < x_2$ and $y \in \mathcal{Y}$ be arbitrary. We have
   \begin{align}
     \frac12 \bigg( \frac{f(x_1, y)}{x_1} + \frac{f(x_2, y)}{x_2} \bigg) &= \frac{x_1 + x_2}{2 x_1 x_2} \bigg( \frac{x_2}{x_1+x_2} f(x_1,y) + \frac{x_1}{x_1+x_2} f(x_2, y) \bigg) \\
     &\leq \frac{x_1 + x_2}{2 x_1 x_2} f \bigg( \frac{2 x_1 x_2}{x_1 + x_2}, y \bigg) 
   \end{align}
   by the concavity of $f(\cdot,y)$. Thus, choosing $x_3 = \frac{2 x_1 x_2}{x_1 + x_2} \in [x_1, x_2] \subset \mathcal{X}$, we see that $g(x,y)$ is indeed $\frac{1}{2}$-concavelike according to~\eqref{eq:concavelike}.
\end{proof}

\section{Proof of Lemma~\ref{lm:duality}}
\label{app:dual}

Let us first show that $\widetilde{I}_{\alpha}(\rho_{AB}\|\tau_A) = -\widetilde{I}_{\beta}(\rho_{AC}\|\tau_A^{-1})$ for pure states $\rho_{ABC}$ and $\frac{1}{\alpha} + \frac{1}{\beta} = 2$. By symmetry it is sufficient to prove the statement for $\alpha \in \big[\frac12, 1\big)$ guaranteeing that $\beta > 1$.

\begin{proof}[{Proof of Eq.~\eqref{eq:duality1} in Lemma~\ref{lm:duality}}]
   First, recall that
  \begin{align}
     \exp \big( -\gamma\, \widetilde{I}_{\alpha}(\rho_{AB}\|\tau_A) \big) &= 
     \sup_{\sigma_B \in \cS(B)} \Big\| \big( \tau_A^{\gamma} \otimes \sigma_B^{\gamma} \big) \times \rho_{AB}  \Big\|_{\alpha} \label{eq:new-rep1}
  \end{align}
  where we set $\gamma := \frac{1-\alpha}{\alpha} \in (0, 1]$ and we use the shorthand $\|X\|_{\alpha} := (\tr (X^{\alpha}))^{\frac{1}{\alpha}}$ for Hermitian operators $X$. Also recall that we use $L \times R$ to denote the Hermitian operator $\sqrt{L} R \sqrt{L}$.  
  
  By introducing the purification $\rho_{ABC}$, we see that
  \begin{align}
    \sup_{\sigma_B \in \cS(B)} \Big\| \big( \tau_A^{\gamma} \otimes \sigma_B^{\gamma} \big) \times \rho_{AB}  \Big\|_{\alpha} 
    &= \sup_{\sigma_B \in \cS(B)}  \Big\| \tr_C \Big[ \big( \tau_A^{\gamma} \otimes \sigma_B^{\gamma} \otimes \id_C \big) \times \rho_{ABC} \Big] \Big\|_{\alpha} \\
    &= \sup_{\sigma_B \in \cS(B)}  \Big\| \tr_{AB} \Big[ \big( \tau_A^{\gamma} \otimes \sigma_B^{\gamma} \otimes \id_C \big) \times \rho_{ABC} \Big] \Big\|_{\alpha} \\
    &= \sup_{\sigma_B \in \cS(B)} \inf_{\sigma_C \in \cS(C) \atop \sigma_C > 0} \tr \Big[ \sigma_C^{-\gamma} \, \tr_{AB} \Big[ \big( \tau_A^{\gamma} \otimes \sigma_B^{\gamma} \otimes \id_C \big) \times \rho_{ABC} \Big] \Big] \label{eq:use-freds-lemma} \\ 
    &= \sup_{\sigma_B \in \cS(B)} \inf_{\sigma_C \in \cS(C) \atop \sigma_C > 0} \tr \Big[ \big( \tau_A^{\gamma} \otimes \sigma_B^{\gamma} \otimes \sigma_C^{-\gamma} \big) \rho_{ABC} \Big] ,
  \end{align} 
  where we employed~\cite[Lm.~12]{lennert13} to establish~\eqref{eq:use-freds-lemma}. Now, it is easy to verify that $\cS(B)$ is convex compact, the set of strictly positive elements of $\cS(C)$ is convex and the function $\tr [ ( \tau_A^{\gamma} \otimes \sigma_B^{\gamma} \otimes \sigma_C^{-\gamma} ) \rho_{ABC} ]$ is concave in $\sigma_B$ and convex in $\sigma_C$. Thus, Sion's minimax theorem~\cite{sion58} applies and yields the following alternative expression:
  \begin{align}
    \exp \big( -\gamma\, \widetilde{I}_{\alpha}(\rho_{AB}\|\tau_A) \big) &=  \inf_{\sigma_C \in \cS(C) \atop \sigma_C > 0} \sup_{\sigma_B \in \cS(B)} \tr \Big[ \big( \tau_A^{\gamma} \otimes \sigma_B^{\gamma} \otimes \sigma_C^{-\gamma} \big) \rho_{ABC} \Big] \\
    &= \inf_{\sigma_C \in \cS(C) \atop \sigma_C > 0} \sup_{\sigma_B \in \cS(B)} \tr \Big[ \sigma_B^{\gamma}\, \tr_{AC} \Big[ \big( \tau_{A}^{\gamma} \otimes \sigma_C^{-\gamma} \big) \times \rho_{ABC}
    \Big] \Big] \\
    &= \inf_{\sigma_C \in \cS(C) \atop \sigma_C > 0} \Big\| \tr_{AC} \Big[ \big( \tau_{A}^{\gamma} \otimes \sigma_C^{-\gamma} \big) \times \rho_{ABC} \Big] \Big\|_{\beta} \label{eq:use-freds-lemma2} \\
    &= \inf_{\sigma_C \in \cS(C) \atop \sigma_C > 0} \Big\| \big( \tau_{A}^{\gamma} \otimes \sigma_C^{-\gamma} \big) \times \rho_{AC} \Big\|_{\beta} \label{eq:new-rep2} .
  \end{align}
  We again used~\cite[Lm.~12]{lennert13} to establish~\eqref{eq:use-freds-lemma2} and note that $\beta = \frac{\alpha}{2\alpha -1} = \frac{1}{1-\gamma}$.
  Substituting for $\gamma = - \frac{1-\beta}{\beta}$ in~\eqref{eq:new-rep2} establishes the desired equality.  
\end{proof}

Let us now show that $I_{\alpha}(\rho_{AB}\|\tau_A) = - \widetilde{D}_{\beta}(\rho_{AC} \| \tau_A^{-1} \otimes \rho_C)$ for pure states $\rho_{ABC}$ and $\alpha \beta = 1$.

\begin{proof}[Proof of Eq.~\eqref{eq:duality2} in Lemma~\ref{lm:duality}]
The following quantum Sibson's identity is adapted from~\cite[Lem.~3 in Suppl.~Mat.]{sharma13}.
Let $\rho_{AB} \in \cS(AB)$, $\tau_A \in \cS(A)$, and $\sigma_B \in \cS(B)$. For any $\alpha > 0$, we have
   \begin{align}
     D_{\alpha}(\rho_{AB}\|\tau_A \otimes \sigma_B) = D_{\alpha}(\rho_{AB}\|\tau_A \otimes \sigma_B^*(\alpha)) + D_{\alpha}(\sigma_B^*(\alpha)\|\sigma_B), \qquad \textrm{where} \label{eq:sibson-def2a} \\
   \sigma_B^*(\alpha) := \frac{ \left( \tr_A \left[ \tau_A^{{1-\alpha}}  \rho_{AB}^{\alpha} \right] \right)^{\frac{1}{\alpha}}}{\tr\Big[\left( \tr_A \left[ \tau_A^{{1-\alpha}}  \rho_{AB}^{\alpha} \right] \right)^{\frac{1}{\alpha}} \Big]} \ . \label{eq:sibson-def2}
   \end{align}

Furthermore, as an immediate consequence of the positive definiteness of $D_{\alpha}(\sigma_B^*(\alpha)\|\sigma_B)$, we find that $\argmin_{\sigma_B \in \cS(B)} D_{\alpha}(\rho_{AB} \| \tau_A \otimes \sigma_B) = \sigma_B^*(\alpha)$ is unique.

In particular, the following chain of equalities holds for any purification $\rho_{ABC}$ of $\rho_{AB}$:
\begin{align}
  I_{\alpha}(\rho_{AB}\|\tau_A) &= \frac{\alpha}{\alpha-1} \log \tr \Big[ \big( \tr_A \big[ \tau_A^{{1-\alpha}} \times \rho_{AB}^{\alpha}  \big] \big)^{\frac{1}{\alpha}} \Big] \\
  &= \frac{\alpha}{\alpha-1} \log \tr \Big[ \big( \tr_{AC} \big[ ( \tau_A^{{1-\alpha}} \otimes \rho_{AB}^{{\alpha-1}}) \times \rho_{ABC} \big] \big)^{\frac{1}{\alpha}} \Big] \\
  &= \frac{\alpha}{\alpha-1} \log \tr \Big[ \big( \tr_{AC} \big[ ( \tau_A^{{1-\alpha}} \otimes \rho_C^{{\alpha-1}}) \times \rho_{ABC} \big] \big)^{\frac{1}{\alpha}} \Big] \\
 &= \frac{\alpha}{\alpha-1} \log \tr \Big[ \big( ( \tau_A^{{1-\alpha}} \otimes \rho_C^{{\alpha-1}}) \times \rho_{AC} \big)^{\frac{1}{\alpha}} \Big] \label{eq:dual2-spectral} \\ 
  &= -\frac{1}{\beta-1} \log \tr \bigg[ \bigg( \Big( \tau_A^{-\frac{1-\beta}{\beta}} \otimes \rho_C^{\frac{1-\beta}{\beta}}\Big) \times \rho_{AC} \bigg)^{\beta} \bigg] = - \widetilde{D}_{\beta}(\rho_{AC} \| \tau_A^{-1} \otimes \rho_C) \,. \label{eq:dual2-last}
\end{align}
 To establish~\eqref{eq:dual2-spectral} we used that the marginals on $AC$ and $B$ of the pure state $\big( \tau_A^{{1-\alpha}} \otimes \rho_C^{{\alpha-1}}\big) \times \rho_{ABC}$ have the same eigenvalues. Finally, we substituted $\beta = \frac1{\alpha}$ in~\eqref{eq:dual2-last}.
\end{proof}

\begin{proof}[Proof of Eq.~\eqref{eq:duality3} in Lemma~\ref{lm:duality}]

We choose a vector $|\psi\rangle$ on the joint system such that $\rho_{ABC}=|\psi\rangle\langle \psi|$. 
Then, using $s:=\alpha-1$, we find the following chain of equalities:
\begin{align*}
s D_{1+s}(\rho_{AB} \| \tau_A \otimes\rho_B)
&= \log \tr \big[\rho_{AB}^{1+s}(\tau_A^{-s} \otimes \rho_B^{-s}) \big]\\
&=\log \tr \big[ \rho_{AB} \rho_{AB}^{s}(\tau_A^{-s} \otimes \rho_B^{-s}) \big]\\
&=\log\, \langle \psi| (\rho_{AB}^{s}\otimes 1_C) (\tau_A^{-s}\otimes I_{BC}) (1_{AC} \otimes \rho_B^{-s}) |\psi\rangle \\
&= \log\, \langle \psi| (1_{AB} \otimes \rho_{C}^{s}) (\tau_A^{-s}\otimes 1_{BC})(\rho_{AC}^{-s} \otimes 1_B)
|\psi\rangle \\
&= \log \tr \big[\rho_{AC} (1_{A} \otimes \rho_{C}^{s}) (\tau_A^{-s}\otimes 1_{C}) \rho_{AC}^{-s} \big]\\
&= \log \tr \big[\rho_{AC}^{1-s} (\tau_A^{-s} \otimes \rho_{C}^{s}) \big] \\
&= -s D_{1-s} \big(\rho_{AC} \,\big\| \tau_A^{-1} \otimes\rho_C \big) \,.
\end{align*}
\end{proof}

\section{Characterization of Minimizers}
\label{app:diff}

\subsection{Fr\'echet Derivatives}

We use the following Fr\'echet derivatives, which seem to be useful for our purposes because they keep us inside the space of normalized density operators. 
For two density operators $\sigma, \omega \in \cS(A)$ and a map $F$ from $\cS(A)$ to $\cP(B)$, define
\begin{align}
  \partial_{\omega} F(\sigma) := D F(\sigma)[\omega - \sigma]  \quad \textrm{and} \quad \partial_{\omega} \partial_{\omega'} F(\sigma) := D^2 F(\sigma)[\omega - \sigma, \omega' - \sigma]
\end{align}
if the first Fr\'echet derivative, $DF(\sigma)$, and the second Fr\'echet derivative, $D^2 F(\sigma)$, exist. In this case we call $F$ differentiable and twice differentiable, respectively.
By linearity of the Fr\'echet derivative we have $\partial_{\mu \omega_1 + (1-\mu) \omega_2} = \mu \partial_{\omega_1} + (1-\mu) \partial_{\omega_2}$ for $\mu \in [0,1]$. Moreover, the second Fr\'echet derivative is symmetric, and hence $\partial_{\omega} \partial_{\omega'} = \partial_{\omega'} \partial_{\omega}$.

We will now summarize some properties of these derivatives. (See~\cite[Sec.~V.3 and Sec.~X.4]{bhatia97} for definitions and an introduction to matrix Fr\'echet derivatives.)
First, if the Fr\'echet derivatives exist, we can also write them as directional derivatives
\begin{align}
\partial_{\omega} F(\sigma) = \frac{\partial}{\partial s} \Big|_{s=0} F\big( (1-s) \sigma + s \omega\big)  \quad \textrm{and} \quad
\partial_{\omega}^2 F(\sigma) = \frac{\partial^2}{\partial s^2} \Big|_{s=0} F\big( (1-s) \sigma + s \omega\big) \,. 
\label{eq:second-deriv-def}
\end{align}
For example, $\partial_{\omega} \sigma = \omega - \sigma$ and $\partial_{\omega}^2 \sigma = 0$.
The derivatives satisfy the usual rules of differentiation. In particular, if both $F$ and $G$ are twice differentiable, we have the chain rules (see, e.g.,~\cite[p.~316]{bhatia97})
\begin{align}
  \partial_{\omega}(G \circ F)(\sigma) &= D G (F(\sigma)) \big[\partial_{\omega}F(\sigma)\big], \label{eq:chain1} \qquad \textrm{and}\\
  \partial_{\omega}^2(G \circ F)(\sigma) &= D G (F(\sigma))\big[\partial_{\omega}^2 F(\sigma) \big] + D^2 G (F(\sigma)) \big[\partial_{\omega}F(\sigma), \partial_{\omega}F(\sigma) \big]\,. \label{eq:chain2}
\end{align}
Hence, in particular $\partial_{\omega}(\tr \circ\, F)(\sigma) = \tr(\partial_{\omega} F(\sigma))$ and $\partial_{\omega}^2(\tr \circ\, F)(\sigma) = \tr(\partial_{\omega}^2 F(\sigma))$.

We often consider smooth functions $g: \mathbb{R}^+ \setminus \{0\} \to \mathbb{R}^+$ defined on the positive real axis that are lifted to positive operators. As a consequence of~\cite[Lem.~V.3.1]{bhatia97} we can write the first Fr\'echet derivative as a Hadamard product, i.e. $Dg(\sigma)(A) = g^{[1]}(\sigma) \odot A$, where $g^{[1]}$ is the matrix of divided differences and $\odot$ denotes the Hadamard product in an eigenbasis of $\sigma > 0$. Here we only need to know that the diagonal elements of $g^{[1]}(\sigma)$ correspond to the eigenvalues of $g'(\sigma)$.
Taking the trace it is then easy to verify that
\begin{align}
  \tr\big[Dg(\sigma)[A]\big] = \tr\big[g'(\sigma) A\big] \quad \textrm{and, in fact,} \quad \tr\big[B\,Dg(\sigma)(A)\big] = \tr\big[B g'(\sigma) A\big] \,. \label{eq:tiawn}
\end{align}
for any $B$ that commutes with $\sigma$. The first statement and the chain rules~\eqref{eq:chain1}--\eqref{eq:chain2} imply that
\begin{align}
  \partial_{\omega}  (\tr \circ\, g \circ F)(\sigma) &= \tr \big[ g'(F(\sigma))\, \partial_{\omega} F(\sigma) \big] 
  \label{eq:diff-trace} \,,\\
    \partial_{\omega}^2 (\tr \circ\, g \circ F)(\sigma)  &= \tr \big[ g'(F(\sigma))\, \partial_{\omega}^2 F(\sigma) \big] + D^2 (\tr \circ g)(F(\sigma))\big[ \partial_{\omega}F(\sigma), \partial_{\omega}F(\sigma) \big] 
\,, \label{eq:diff2-trace}
\end{align}
as long as $F(\sigma) > 0$ or $g$ is differentiable at zero. In case $g'(0)$ is undefined but $F(\sigma) \gg \partial_{\omega}F(\sigma)$ the equality~\eqref{eq:diff-trace} still holds by a continuity argument if we set $g'(0) = 0$.

Finally, we recall that the trace functional $\tr \circ\, g$ from $\cS(A)$ to $\mathbb{R}$ is (strictly) convex if $g$ is (strictly) convex and (strictly) concave if $g$ is (strictly) concave (see, e.g.,~\cite[Thm.~2.10]{carlen09}). Furthermore, if $g$ is twice continuously differentiable with $g'' > 0$, then we deduce from~\eqref{eq:diff2-trace} that
\begin{align}
  \partial_{\omega}^2 (\tr \circ\,g \circ F)(\sigma))  &> \tr \big[ g'(F(\sigma)) \partial_{\omega}^2 F(\sigma) \big] \,, \label{eq:strictly}
\end{align}
unless $\partial_{\omega}F(\sigma) = 0$, 
and the same holds with opposite sign if $g'' < 0$. Finally, due to the representation as a directional derivative in~\eqref{eq:second-deriv-def}, it is evident that if $g$ is operator convex, then $\partial_{\omega}^2 g(\sigma) \geq 0$ and similarly if $g$ is operator concave, then $\partial_{\omega}^2 g(\sigma) \leq 0$.

\subsection{Characterization of the Optimal Marginal State} 
\label{sc:opt-marginal}

In this section we provide the proof of Lemma~\ref{pr:fixed-point-lemma}. The argument is split into two parts, presented in Lemmas~\ref{lm:unique-max1} and~\ref{lm:unique-max2}, which together imply Lemma~\ref{pr:fixed-point-lemma}. 

Before we start, note that Lemma~\ref{lm:iso} allows us to restrict our attention to $\rho_{AB} \in \cS(AB)$ is such that both $\rho_A$ and $\rho_B$ and, thus, $\tau_A \in \cS(A)$ have full support. 
Consider the functional
  \begin{align}
   \chi_{\alpha}: \sigma_B \mapsto \tr \Big[ \big( (\tau_A \otimes \sigma_B)^{\frac{1-\alpha}{2\alpha}} \rho_{AB} (\tau_A \otimes \sigma_B)^{\frac{1-\alpha}{2\alpha}}\big)^{\alpha} \Big] 
  \end{align}
for $\alpha \geq \frac12$ which corresponds to the trace functionals used in the definition of $\tilde{I}_{\alpha}(\rho_{AB}\|\tau_A)$. Now let us first consider the case $\alpha \in \big[\frac12, 1)$. We see that the directional derivative on the boundary of $\cS(B)$ where at least one eigenvalue is zero in a direction that increases the rank diverges to positive infinity. Since $\chi_{\alpha}$ is continuous on any compact subset and the boundary of $\cS(B)$ is compact, this allows us to conclude that
\begin{align}
  \sup_{\sigma_B \in \cS(B)} \chi_{\alpha}(\sigma_B) = \max_{\sigma_B \in \cS_{\nu}(B)} \chi_{\alpha}(\sigma_B)
\end{align}
where $\cS_{\nu}(B) := \{ \sigma_B \in \cS(B) : \sigma_B \geq \nu 1_B \}$ for some small constant $\nu > 0$ that depends on the states $\rho_{AB}$ and $\tau_B$ as well as $\alpha$. This compact set contains states with eigenvalues bounded away from zero.
  Next, we consider $\alpha > 1$. Since $\chi_{\alpha}(\sigma_B)$ diverges to positive infinity when an eigenvalue of $\sigma_B$ approaches zero, we can again write
\begin{align}
  \inf_{\sigma_B \in \cS(B)} \chi_{\alpha}(\sigma_B) = \min_{\sigma_B \in \cS_{\nu}(B)} \chi_{\alpha}(\sigma_B)
\end{align}
for an appropriate choice of $\nu > 0$.
 Using the above relations, we define the set of states that achieve the maximum (for $\alpha < 1$) or minimum (for $\alpha > 1$) as follows:
\begin{align}
  \mathcal{M}_{\alpha}(B) := \begin{cases}  \argmax_{\sigma_B \in \cS(B)} \chi_{\alpha}(\sigma_B) & \textrm{if } \alpha \in \big[\frac12, 1\big) \\ 
    \argmin_{\sigma_B \in \cS(B)} \chi_{\alpha}(\sigma_B) & \textrm{if } \alpha \in (1, \infty) \end{cases}
    \label{eq:theopt}
\end{align}
and note that they are non-empty.

Next we characterize the states that achieve the optimum. 
\begin{lemma}
\label{lm:unique-max1}
Let $\rho_{AB} \in \cS(AB)$ and $\tau_A \in \cS(A)$ such that $\tau_A \gg \rho_A$, and $\alpha \in [\frac12,1) \cup (1,\infty)$.   Moreover, let $\cF_{\alpha}(B) \subset 
\{ \sigma_B \in \cS(B) : \sigma_B \gg \rho_B \}$ be the fixed-points of the following non-linear map:
  \begin{align}
  \mathcal{X}_{\alpha} : \sigma_B \mapsto \frac{1}{\chi_{\alpha}(\sigma_B)} \tr_A \Big[ \Big( (\tau_A \otimes \sigma_B)^{\frac{1-\alpha}{2\alpha}} \rho_{AB} (\tau_A \otimes \sigma_B)^{\frac{1-\alpha}{2\alpha}}\Big)^{\alpha} \Big]  \label{eq:themap}
  \end{align}
  Then, $\mathcal{M}_{\alpha}(B) = \cF_{\alpha}(B)$.
\end{lemma}

We stated the lemma for general states but note that in the proof we can readily restrict our attention to states $\rho_{AB}$ such that $\rho_A$ and $\rho_B$ have full support.


\begin{proof}
We derive a necessary and sufficient condition for a state $\sigma_B > 0$ to be an extremum of the optimizations in~\eqref{eq:theopt} as follows. Let $\alpha \in [\frac12,1) \cup (1, \infty)$ and set $\gamma = \frac{1-\alpha}{\alpha} \in (-1, 1] \setminus \{0\}$. Moreover, 
\begin{align}
  X = \tau^{\gamma/2} \rho^{1/2} , \qquad \textrm{such that} \qquad \chi_{\alpha}(\sigma) = \tr \big[ \big( X^{\dagger} \sigma^{\gamma} X \big)^{\alpha} \big] ,
\end{align}
where we omitted the identity symbol and dropped the subscripts to make the presentation more concise in the following. 
Since $\chi_{\alpha}$ is either concave (for $\alpha < 1$) or convex (for $\alpha > 1$)~\cite{frank13} and the optimum is taken in the interior, a necessary and sufficient condition for $\sigma$ to be an optimizer (in either case) is given if $\partial_{\omega}\chi_{\alpha}(\sigma) = 0$ for all $\omega \in \cS(B)$.

Using the relation in~\eqref{eq:diff-trace}, we find the following chain of equalities:
\begin{align}
  \partial_{\omega}\chi_{\alpha}(\sigma) &= \alpha \tr \Big[ \big( X^{\dagger} \sigma^{\gamma} X \big)^{\alpha-1} \cdot \partial_{\omega} \big( X^{\dagger} \sigma^{\gamma} X \big) \Big] \\
  &= \alpha \tr \Big[\sigma^{\gamma/2} X \big( X^{\dagger} \sigma^{\gamma} X \big)^{\alpha-1} X^{\dagger} \sigma^{\gamma/2} \cdot \sigma^{-\gamma/2}  \big( \partial_{\omega}  \sigma^{\gamma} \big) \sigma^{-\gamma/2}\Big] \\
  &= \alpha \tr \Big[ \big( \sigma^{\gamma/2} X X^{\dagger} \sigma^{\gamma/2} \big)^{\alpha} \cdot \sigma^{-\gamma/2}  \big( \partial_{\omega}  \sigma^{\gamma} \big) \sigma^{-\gamma/2} \Big] \label{eq:use-identity1}\\
  &= \alpha \tr \Big[  \sigma^{-1/2} \tr_A \Big[ \big( \sigma^{\gamma/2} X X^{\dagger} \sigma^{\gamma/2} \big)^{\alpha} \Big] \sigma^{-1/2} \cdot \sigma^{(1-\gamma)/2}  \big( \partial_{\omega}  \sigma^{\gamma} \big) \sigma^{(1-\gamma)/2} \Big] \,.\label{eq:simplify-this}
\end{align}
To derive~\eqref{eq:use-identity1} we used the identity $f(Y^{\dag} Y) Y^{\dag} = Y^{\dag} f(Y Y^{\dag})$ for $Y = \sigma^{\gamma/2} X$, which can be verified using the polar decomposition for any linear operator $Y$.

Next we show that the operators
\begin{align}
  \Big\{ \Delta_{\omega} = \sigma^{(1-\gamma)/2}  \big( \partial_{\omega}  \sigma^{\gamma} \big) \sigma^{(1-\gamma)/2} : \omega \in \cS(B) \Big\}
\end{align}
span the space of traceless Hermitian operators on $\mathcal{H}_B$. Introducing the eigenvalue decomposition $\sigma = \sum_x \lambda_x |x\rangle\!\langle x|$ with $\lambda_x > 0$, we can easily verify using~\cite[Thm.~3.25]{hiaipetz14} that 
\begin{align}
  \big\langle x \big| \Delta_{\omega} \big| y \big\rangle = \begin{cases}
  ( \lambda_x \lambda_y )^{\frac{1-\gamma}{2}} \frac{\lambda_x^{\gamma} - \lambda_y^{\gamma}}{\lambda_x - \lambda_y} \langle x | \omega - \sigma | y \rangle & \textrm{if $\lambda_x \neq \lambda_y$} \\
   \gamma\, \langle x | \omega - \sigma | y \rangle & \textrm{if $\lambda_x = \lambda_y$}
  \end{cases} \,.
\end{align}
Hence, $\Delta_{\omega}$ is Hermitian and $\tr[\Delta_{\omega}] = 0$ for all $\omega \in \cS(B)$. Next, note that a basis of the traceless Hermitian operators is given by the operators 
\begin{align}
 \big\{ \Gamma_{xy} = |x \rangle\!\langle y| + |y \rangle\!\langle x|,\, \Gamma_{xy}' = i|x \rangle\!\langle y| -i|y \rangle\!\langle x|,\, \Gamma_{xy}'' = |x \rangle\!\langle x| - |y \rangle\!\langle y| \big\}_{x \neq y} \,.
\end{align}

Furthermore, for every tuple $(x,y)$ with $x \neq y$ there exists an $\eps > 0$ such that the state $\omega = \sigma + \eps \Gamma_{xy}$ is still in $\cS(B)$. For this state, we find that $\Delta_{\omega} = \eta \Gamma_{xy}$ for some real $\eta > 0$. An analogous argument applies to $\Gamma_{xy}'$ and $\Gamma_{xy}''$. Hence, we have verified that
the operators $\{ \Delta_{\omega} \}_{\omega \in \cS(B)}$ span the space of traceless Hermitian operators.

Now let us return to~\eqref{eq:simplify-this}. In light of the above, the condition that $\partial_{\omega}\chi_{\alpha}(\sigma) = 0$ for all $\omega \in \cS(B)$ translates to the condition that the operator
\begin{align}
  \sigma^{-1/2} \tr_A \Big[ \big( \sigma^{\gamma/2} X X^{\dagger} \sigma^{\gamma/2} \big)^{\alpha} \Big] \sigma^{-1/2}
\end{align}
must be proportional to the identity. Hence, the optimum must indeed be a fixed point of the map given in~\eqref{eq:themap}.
\end{proof}

The following lemma implies that $\mathcal{M}_{\alpha}(B)$ contains exactly one element for $\alpha \geq \frac12$.
In light of Lemma~\ref{lm:unique-max1}, this then implies that if such a fixed point exists then we must have $\mathcal{M}_{\alpha}(B) = \cF_{\alpha}(B)$ and, hence, both sets must contain exactly one element.

\begin{lemma}
  \label{lm:unique-max2}
  The map $\sigma_B \mapsto \chi_{\alpha}(\sigma_B)$ is strictly concave with negative definite Hessian for $\alpha \in [\frac12,1)$ and strictly convex with positive definite Hessian for $\alpha > 1$.
\end{lemma}

\begin{proof}

Let us first focus on the case $\alpha \in [\frac12, 1)$ such that $\gamma = \frac{1-\alpha}{\alpha} \in (0,1]$. It suffices to show that $\partial_{\omega}^2 \chi_{\alpha}(\sigma) < 0$ for all $\omega \neq \sigma$ where $\sigma > 0$. We use the notation of Lemma~\ref{lm:unique-max1}. Due to~\eqref{eq:strictly} and the fact that $t \mapsto t^{\alpha}$ has strictly negative second derivative, we have
\begin{align}
  \partial_{\omega}^2 \chi_{\alpha}(\sigma) &< \alpha \tr\big[ \big( X^{\dagger} \sigma^{\gamma} X \big)^{\alpha-1} \cdot \partial_{\omega}^2 \big( X^{\dagger} \sigma^{\gamma} X \big) \big] \,, \label{eq:plugin}
\end{align}
unless $\partial_{\omega}(X^{\dag} \sigma^{\gamma} X ) = 0$. Since we assume that $\rho_{B}$ and $\tau_A$ have full support, it suffices to convince oneself that 
$\partial_{\omega} \sigma^{\gamma} \neq 0$ for all $\omega \neq \sigma$ in order to verify that~\eqref{eq:plugin} holds. Now note that
\begin{align}
 \partial_{\omega}^2 \big( X^{\dagger} \sigma^{\gamma} X \big) = X^{\dag} \big( \partial_{\omega}^2 \sigma^{\gamma} \big) X \leq 0
\end{align}
since $t \mapsto t^{\gamma}$ is operator concave. Plugging this into~\eqref{eq:plugin}, we conclude that $\partial_{\omega}^2 \chi_{\alpha}(\sigma) < 0$.

The argument for $\alpha > 1$ and $\gamma \in (-1, 0)$ proceeds analogously by noting that $t \mapsto t^{\alpha}$ has strictly positive second derivative and $t \mapsto t^{\gamma}$ is operator convex. Moreover, we only need to consider states $\sigma_B$ with full support in this case.
\end{proof}

\section{Differentiability of the R\'enyi Mutual Information}
\label{app:diff2}

 We will need the following Lemma:

\begin{lemma} \label{lm:contin}
  Let $\Omega \subset \mathcal{X}$ for a Banach space $\mathcal{X}$ and let $\Gamma$ be a convex open set in $\mathcal{Y} =\mathbb{R}^m$ for $m \in \mathbb{N}$. Moreover, let $f: \Omega' \times \Gamma \to \mathbb{R}$ be twice continuously Fr\'echet differentiable in an open set $\Omega'$ containing $\Omega$ and has a strictly positive Hessian with regards to $\Gamma$ for all $x \in \Omega$. 
  If further $\bar{y}(x) := \argmin_{y \in \Gamma} f(x,y)$ exists in $\Gamma$ for each $x \in \Omega$, then $g(x) := f(x, \bar{y}(x))$ is continuously Fr\'echet differentiable on $\Omega$ with $D g(x) = D_{\mathcal{X}} f(x, y) |_{y=\bar{y}(x)}$, where $D_{\mathcal{X}}$ denotes the partial Fr\'echet derivative with regards to $\mathcal{X}$.
\end{lemma}

\begin{proof}
  Let us fix any $x_0 \in \Omega$. Note that the minimum, $\bar{y}_0 := \argmin_{y \in \Gamma} f(x_0,y)$, is unique due to our assumption of strict convexity. Moreover, the derivative at the minimum satisfies
  \begin{align}
    D_{\mathcal{Y}} f(x_0, \bar{y}_0)  = 0
  \end{align}
  where $D_{\mathcal{Y}}$ denotes the partial Fr\'echet derivative with regards to $\mathcal{Y}$. Now pick a basis $\{ a_i \}_{i \in [m]}$ of $\mathcal{Y}$ and define $F_i(x,y) := D_{\mathcal{Y}} f(x, y) [a_i]$. We want to apply the implicit function theorem to the vector valued map $F: \Omega' \times \mathcal{Y} \to \mathcal{Y}$. By assumption $F$ is continuously differentiable in an open neighborhood of $(x_0, \bar{y}_0)$. The derivative with regards to $\mathcal{Y}$ is invertible since the Jacobian of $F$ corresponds to the Hessian of $f$ (with regards to $\mathcal{Y}$), i.e.\
  \begin{align}
    (J)_{i,j} = D_{\mathcal{Y}} F_i(x_0,\bar{y}_0) [a_j]  = D_{\mathcal{Y}}^2 f(x_0, \bar{y}_0) [a_i, a_j]  ,
  \end{align}
  and the latter is positive definite by assumption. Hence, $\bar{y}(x)$ is continuously differentiable in some open neighborhood of $x_0$ and
  \begin{align}
    D g(x) = D_{\mathcal{X}} f(x, \bar{y}(x)) &= D_{\mathcal{X}} f(x, y) \big|_{y = \bar{y}(x)} + D_{\mathcal{Y}} f(x,y) [ D_{\mathcal{X}} \bar{y}(x) ] \Big|_{y = \bar{y}(x)} \\
    &= D_{\mathcal{X}} f(x, y) \big|_{y = \bar{y}(x)}
  \end{align}
  by the law of total differentiation.
  Since this holds for any $x_0 \in \Omega$, we conclude the proof.
\end{proof}

This now allows us to show Proposition~\ref{pr:diff}, which we restate here for the convenience of the reader.

\newtheorem*{prop11}{Proposition 11 (restated)}

\begin{prop11}
  Let $\rho_{AB} \in \cS(AB)$ and $\tau_A \in \cS(A)$ such that $\tau_A \gg \rho_A$. Then, the function $\alpha \mapsto \widetilde{I}_{\alpha}(\rho_{AB}\|\tau_A)$ is continuously differentiable for $\alpha \geq \frac12$ with
  \begin{align}
      \frac{\mathrm{d}}{\mathrm{d} \alpha} \widetilde{I}_{\alpha}(\rho_{AB}\|\tau_A) = \frac{\partial}{\partial \alpha} \widetilde{D}_{\alpha}(\rho_{AB}\|\tau_A \otimes \sigma_B) \Big|_{\sigma_B = \widetilde{\sigma}_B^*(\alpha)} ,
  \end{align}
  where $\widetilde{\sigma}_B^*(\alpha)$ is the optimizer in Lemma~\ref{pr:fixed-point-lemma}.
  In particular,
 $\frac{\mathrm{d} }{\mathrm{d}  \alpha} \widetilde{I}_{\alpha}(\rho_{AB}\|\tau_A) \big|_{\alpha = 1} = \frac{1}{2 \log (e)} V(\rho_{AB} \| \tau_A)$.
\end{prop11}

\begin{proof}
We first show that $\alpha \mapsto \widetilde{I}_{\alpha}(\rho_{AB}\|\tau_A)$ is continuously differentiable for $\alpha \in \big[\frac12, 1) \cup (1, \infty)$. Note that the set of strictly positive definite operators is an open set in the convex set of Hermitian operators with unit trace. 
As we have argued in Section~\ref{sc:opt-marginal}, $\chi_{\alpha}(\sigma_B)$ takes its optimum in this set for all values of $\alpha$ under consideration. Clearly, $(\alpha, \sigma) \mapsto \chi_{\alpha}(\sigma)$ is twice continuously Fr\'echet differentiable. Moreover, according to Lemma~\ref{lm:unique-max2}, the function $\sigma \mapsto \chi_{\alpha}(\sigma)$ has positive definite Hessian (for $\alpha > 1$) or negative definite Hessian (for $\alpha < 1$). Hence, Lemma~\ref{lm:contin} establishes the desired continuous differentiability.

It remains to consider the limit $\alpha \to 1$. Let us first calculate the derivative at $\alpha =1$. Note that
\begin{align}
  &\limsup_{h \to 0} \left\{ \frac{1}{h} \left( \widetilde{I}_{1+h}(\rho_{AB}\|\tau_A) - I(\rho_{AB}\|\tau_A) \right) \right\} \\
  &\qquad \qquad \leq \limsup_{h \to 0} \left\{  \frac{1}{h} \left( \widetilde{D}_{1+h}(\rho_{AB}\|\tau_A \otimes \rho_B) - D(\rho_{AB}\|\tau_A \otimes \rho_B) \right) \right\} \\
  &\qquad \qquad= \frac{\partial}{\partial h} \widetilde{D}_{1+h}(\rho_{AB}\|\tau_A \otimes \rho_B)  \bigg|_{h = 0} = \frac{1}{2 \log (e)} V(\rho_{AB}\|\tau_A \otimes \rho_B) \,.
\end{align}
  On the other hand, using Lemma~\ref{lm:duality}, we find
  \begin{align}
&\widetilde{I}_{1+h}(\rho_{AB}\|\tau_A) - I(\rho_{AB}\|\tau_A) = I(\rho_{AC}\|\tau_A^{-1}) - \widetilde{I}_{1-f(h)}(\rho_{AC}\|\tau_A^{-1}) ,
  \end{align}
  where $f: h \mapsto \frac{h}{1+2h}$ satisfies $f(0) = 0$ and $f'(0) = 1$. Using this, we find
\begin{align}
&\liminf_{h \to 0} \left\{ \frac{1}{h} \left( \widetilde{I}_{1+h}(\rho_{AB}\|\tau_A) - I(\rho_{AB}\|\tau_A) \right) \right\} \\
&\qquad \qquad = \liminf_{h \to 0} \left\{ \frac{1}{h} \left( I(\rho_{AC}\|\tau_A^{-1}) - \widetilde{I}_{1-f(h)}(\rho_{AC}\|\tau_A^{-1}) \right) \right\} \\
&\qquad \qquad \geq \liminf_{h \to 0} \left\{ \frac{1}{h} \left( D(\rho_{AC}\|\tau_A^{-1} \otimes \rho_{C}) - \widetilde{D}_{1-f(h)}(\rho_{AC}\|\tau_A^{-1} \otimes \rho_C) \right) \right\} \\
&\qquad \qquad = - \frac{\partial}{\partial h} \widetilde{D}_{1-f(h)}(\rho_{AC}\|\tau_A^{-1} \otimes \rho_C) \bigg|_{h=0} =  \frac{1}{2 \log (e)} V(\rho_{AC}\|\tau_A^{-1} \otimes \rho_C) ,
\end{align}
where the final equation follows from \eqref{H-28}.
Employing \eqref{eq:duality4} establishes the derivative at $\alpha = 1$.

In order to show that the derivative is continuous at $\alpha = 1$, it remains to show that $\widetilde{\sigma}_B^*(\alpha)$ converges to $\widetilde{\sigma}_B^*(1) = \rho_B$ in the limit $\alpha \to 1$. According to Lemma~\ref{pr:fixed-point-lemma}, the optimizer $\widetilde{\sigma}_B^*(\alpha)$ is the (unique) solution with full rank to
 \begin{align}
   f(\alpha, \sigma_B) = \mathcal{X}_{\alpha}(\sigma_B) - \chi_{\alpha}(\sigma_B) \sigma_B = 0
 \end{align}
 Clearly $\sigma_B \mapsto f(1, \sigma_B) = \rho_B - \sigma_B$ has an invertible Jacobian, and thus the implicit function theorem ensures that $\widetilde{\sigma}_B^*(\alpha)$ is continuous around $\alpha = 1$.
\end{proof}

\bibliographystyle{arxiv}
\bibliography{library}

\end{document}